\documentclass[%
superscriptaddress,
nofootinbib,
nobibnotes,
amsmath,amssymb,
aps,
prx,
onecolumn,
floatfix,
longbibliography
]{revtex4-2}
\usepackage[english]{babel}
\usepackage[T1]{fontenc}
\usepackage{graphicx}
\usepackage{float}

\makeatletter
\let\newfloat\newfloat@ltx 
\makeatother

\usepackage{algorithm}%
\usepackage{algpseudocode}%

\usepackage{todonotes}
\usepackage{hyperref}
\usepackage[capitalise]{cleveref}

\usepackage{qcircuit}
\usepackage{float}
\usepackage{mathtools}
\usepackage{hyperref}
\usepackage{fancyhdr}
\usepackage{extramarks}
\usepackage{amsmath}
\usepackage{amsthm}
\usepackage{amsfonts}
\usepackage{amssymb}
\usepackage{tikz}
\usepackage{dsfont}
\usepackage{bbm}
\usepackage{diagbox}
\usepackage{wasysym}

\usepackage{amssymb}
\usepackage{verbatim}
\usepackage{tabularx,lipsum,environ}
\usepackage{csquotes}

\newtheorem{theorem}{Theorem}

\newtheorem{lemma}[theorem]{Lemma}
\newtheorem{definition}{Definition}
\newtheorem{question}{Question}

\newcommand{\naturel}{\mathbb{N}}

\usepackage{bm}
\usepackage{bbm}
\usepackage{thmtools}

\newcommand{\sm}{\setminus}
\newcommand{\set}[1]{\left\{#1\right\}}

\newcommand{\rw}{\mathsf{rw}}
\newcommand{\tw}{\mathsf{tw}}
\newcommand{\frw}[1]{\mathsf{frw}_#1}
\newcommand{\ftw}[1]{\mathsf{ftw}_#1}

\usepackage{appendix}

\usepackage{tikzit}

\tikzstyle{gate}=[shape=rectangle, text height=1.5ex, text depth=0.25ex, yshift=0.5mm, fill=white, draw=black, minimum height=3mm, yshift=-0.5mm, minimum width=3mm, font={\small}, tikzit category=circuit, inner sep=2pt]
\tikzstyle{big gate}=[shape=rectangle, text height=1.5ex, text depth=0.25ex, yshift=0.5mm, fill=white, draw=black, minimum height=10mm, yshift=-0.5mm, minimum width=5mm, font={\small}, tikzit category=circuit]
\tikzstyle{Z dot}=[inner sep=0mm, minimum size=2mm, shape=circle, draw=black, fill={rgb,255: red,221; green,255; blue,221}, tikzit category=zx]
\tikzstyle{Z phase dot}=[minimum size=5mm, font={\footnotesize\boldmath}, shape=rectangle, rounded corners=2mm, inner sep=0.2mm, outer sep=-2mm, scale=0.8, tikzit shape=circle, draw=black, fill={rgb,255: red,221; green,255; blue,221}, tikzit draw=blue, tikzit category=zx]
\tikzstyle{X dot}=[Z dot, shape=circle, draw=black, fill={rgb,255: red,255; green,136; blue,136}, tikzit category=zx]
\tikzstyle{X phase dot}=[Z phase dot, tikzit shape=circle, tikzit draw=blue, fill={rgb,255: red,255; green,136; blue,136}, font={\footnotesize\boldmath}, tikzit category=zx]
\tikzstyle{hadamard}=[fill=yellow, draw=black, shape=rectangle, inner sep=0.6mm, minimum height=1.5mm, minimum width=1.5mm, tikzit category=zx]
\tikzstyle{paulibox}=[fill={rgb,255: red,221; green,221; blue,255}, draw=black, shape=rectangle, inner sep=0.6mm, minimum height=5mm, minimum width=5mm, font={\footnotesize}, text height=1.5ex, text depth=0.25ex, tikzit category=zx]
\tikzstyle{vertex}=[inner sep=0mm, minimum size=1mm, shape=circle, draw=black, fill=black, tikzit category=misc]
\tikzstyle{vertex set}=[inner sep=0mm, minimum size=1mm, shape=circle, draw=black, fill=white, font={\footnotesize\boldmath}, tikzit category=misc]
\tikzstyle{small black dot}=[fill=black, draw=black, shape=circle, inner sep=0pt, minimum width=1.2mm, tikzit category=circuit]
\tikzstyle{cnot ctrl}=[fill=black, draw=black, shape=circle, inner sep=0pt, minimum width=1.2mm, tikzit category=circuit]
\tikzstyle{cnot targ}=[fill=white, draw=white, shape=circle, tikzit category=circuit, label={center:$\oplus$}, inner sep=0pt, minimum width=2.1mm, tikzit fill={rgb,255: red,102; green,204; blue,255}, tikzit draw=black]
\tikzstyle{ket}=[fill=white, draw=black, shape=regular polygon, regular polygon sides=3, regular polygon rotate=-30, scale=0.7, inner sep=1pt, tikzit category=circuit, tikzit shape=rectangle, tikzit fill=green]
\tikzstyle{bra}=[fill=white, draw=black, shape=regular polygon, regular polygon sides=3, regular polygon rotate=30, scale=0.7, inner sep=1pt, tikzit category=circuit, tikzit shape=rectangle, tikzit fill=red]
\tikzstyle{scalar}=[shape=rectangle, text height=1.5ex, text depth=0.25ex, yshift=0.5mm, fill=white, draw=black, minimum height=5mm, yshift=-0.5mm, minimum width=5mm, font={\small}]
\tikzstyle{clabel}=[fill=white, draw=none, shape=rectangle, tikzit fill={rgb,255: red,56; green,255; blue,242}, font={\footnotesize}, inner sep=1pt, tikzit category=labels]
\tikzstyle{empty diagram}=[draw={gray!40!white}, dashed, shape=rectangle, minimum width=1cm, minimum height=1cm, tikzit category=misc]
\tikzstyle{amap}=[fill=white, draw=black, shape=NEbox, tikzit category=asymmetric, tikzit fill=yellow, tikzit shape=rectangle]
\tikzstyle{amap conj}=[fill=white, draw=black, shape=NWbox, tikzit category=asymmetric, tikzit fill=green, tikzit shape=rectangle]
\tikzstyle{amap adj}=[fill=white, draw=black, shape=SEbox, tikzit category=asymmetric, tikzit fill=red, tikzit shape=rectangle]
\tikzstyle{amap trans}=[fill=white, draw=black, shape=SWbox, tikzit category=asymmetric, tikzit fill=orange, tikzit shape=rectangle]
\tikzstyle{astate}=[fill=white, draw=black, shape=NEtriangle, tikzit category=asymmetric, tikzit shape=circle, tikzit fill=yellow]
\tikzstyle{astate conj}=[fill=white, draw=black, shape=NWtriangle, tikzit category=asymmetric, tikzit shape=circle, tikzit fill=green]
\tikzstyle{astate adj}=[fill=white, draw=black, shape=SEtriangle, tikzit category=asymmetric, tikzit shape=circle, tikzit fill=red]
\tikzstyle{astate trans}=[fill=white, draw=black, shape=SWtriangle, tikzit category=asymmetric, tikzit shape=circle, tikzit fill=orange]
\tikzstyle{white dot}=[inner sep=0mm, minimum size=2mm, shape=circle, draw=black, fill={rgb,255: red,250; green,250; blue,250}]
\tikzstyle{white phase dot}=[minimum size=5mm, font={\footnotesize\boldmath}, shape=rectangle, rounded corners=2mm, inner sep=0.2mm, outer sep=-2mm, scale=0.8, tikzit shape=circle, draw=black, fill={rgb,255: red,250; green,250; blue,250}, tikzit draw=blue]
\tikzstyle{hbox}=[shape=rectangle, text height=2mm, fill={rgb,255: red,255; green,235; blue,61}, draw=black, minimum height=2mm, minimum width=2mm, font={\small}, tikzit category=zh, inner sep=0pt, rounded corners=0.5mm]
\tikzstyle{Z dot (zh)}=[inner sep=0mm, minimum size=2mm, shape=circle, draw=black, fill={rgb,255: red,250; green,250; blue,250}, tikzit category=zh]
\tikzstyle{X dot (zh)}=[Z dot, shape=circle, draw=black, fill={rgb,255: red,193; green,193; blue,193}, tikzit category=zh]
\tikzstyle{triangle}=[fill={rgb,255: red,255; green,136; blue,136}, draw=black, shape=isosceles triangle, isosceles triangle apex angle=60, minimum size=2.5mm, inner sep=0mm]
\tikzstyle{labelled hbox}=[shape=rectangle, text height=1.75ex, text depth=0.5ex, fill={rgb,255: red,255; green,235; blue,61}, draw=black, minimum height=3mm, minimum width=4mm, font={\small}, tikzit category=zh, inner sep=1.3pt, rounded corners=0.5mm]
\tikzstyle{Z phase dot (zh)}=[Z phase dot, tikzit shape=circle, tikzit draw=blue, fill={rgb,255: red,250; green,250; blue,250}, font={\footnotesize\boldmath}, tikzit category=zh]
\tikzstyle{X phase dot (zh)}=[Z phase dot, tikzit shape=circle, tikzit draw=blue, fill={rgb,255: red,193; green,193; blue,193}, font={\footnotesize\boldmath}, tikzit category=zh]
\tikzstyle{W node}=[fill=black, draw=black, shape=regular polygon, regular polygon sides=3, minimum size=2mm]
\tikzstyle{Z dot (zw)}=[fill=white, draw=black, shape=circle, minimum width=1.2mm, inner sep=0pt]
\tikzstyle{Z phase dot XL}=[Z phase dot, fill={rgb,255: red,250; green,250; blue,250}, draw=black, shape=circle, tikzit draw={rgb,255: red,191; green,0; blue,64}, tikzit shape=circle, font={\large\boldmath}, inner sep=0.0mm]
\tikzstyle{gn}=[fill=green, draw=black, shape=circle, tikzit category=ZX, tikzit fill=green, tikzit draw=black, tikzit shape=circle, inner sep=2pt]
\tikzstyle{rn}=[fill=red, draw=black, shape=circle, tikzit fill=red, tikzit draw=black, tikzit category=ZX, tikzit shape=circle, inner sep=2pt]
\tikzstyle{divide}=[regular polygon, regular polygon sides=3, shape border rotate=90, draw=black, fill=gray, inner sep=1.6pt, tikzit category=scal, rounded corners=0.8mm]
\tikzstyle{black}=[fill=black, draw=black, shape=circle, tikzit fill=black, tikzit draw=black, tikzit shape=circle, tikzit category=IH, inner sep=2pt]
\tikzstyle{gather}=[fill=gray, draw=black, tikzit category=scal, rounded corners=0.8mm, regular polygon, regular polygon sides=3, shape border rotate=-90, inner sep=1.6pt]
\tikzstyle{ggen}=[fill=white, draw=black, shape=rectangle, rounded corners=2mm, line width=1pt, tikzit draw=red, tikzit category=scal]
\tikzstyle{white}=[fill=white, draw=black, shape=circle, inner sep=2pt, tikzit category=IH]
\tikzstyle{mbox}=[fill=white, draw=black, rounded rectangle, rounded rectangle west arc=none, tikzit category=scal, tikzit shape=rectangle]
\tikzstyle{A}=[fill=white, shape=circle, tikzit category=scal, inner sep=1pt]
\tikzstyle{ggreen}=[fill=green, draw=black, shape=circle, tikzit category=SZX, tikzit fill=green, tikzit draw=black, line width=1pt, inner sep=2pt]
\tikzstyle{gred}=[fill=red, draw=black, shape=circle, rounded corners=2mm, tikzit category=SZX, inner sep=2pt, tikzit fill=red, line width=1pt]
\tikzstyle{ghad}=[fill=yellow, draw=black, shape=rectangle, tikzit category=SZX, tikzit shape=rectangle, tikzit fill=yellow, inner sep=2pt, line width=1pt]
\tikzstyle{boxm}=[fill=white, draw=black, rounded rectangle, tikzit category=scal, tikzit shape=rectangle, rounded rectangle east arc=none]
\tikzstyle{box}=[fill=white, draw=black, shape=rectangle]
\tikzstyle{had}=[fill=yellow, draw=black, shape=rectangle, tikzit category=ZX, tikzit fill=yellow, tikzit draw=black, inner sep=2pt]
\tikzstyle{gwhite}=[fill=white, draw=black, shape=circle, tikzit fill=white, tikzit shape=circle, line width=1 pt, inner sep=2 pt, tikzit draw=red]
\tikzstyle{gblack}=[fill=black, draw=black, shape=circle, tikzit fill=black, tikzit shape=circle, line width=1 pt, inner sep=2 pt, tikzit draw=red]
\tikzstyle{antipode}=[fill=red, draw=black, shape=rectangle, tikzit fill=red, tikzit draw=black, tikzit shape=rectangle, inner sep=2pt]
\tikzstyle{diamond}=[fill=white, draw=black, shape=diamond, inner sep=2pt]
\tikzstyle{mongr}=[fill=green, draw=green, shape=circle, inner sep=2pt]
\tikzstyle{monbl}=[fill=blue, draw=black, shape=circle, inner sep=2pt]
\tikzstyle{bg}=[inner sep=0.7mm, minimum width=0pt, minimum height=0pt, fill=green, draw=white, very thick, shape=circle]
\tikzstyle{br}=[inner sep=0.7mm, minimum width=0pt, minimum height=0pt, fill=red, draw=white, very thick, shape=circle]
\tikzstyle{rmat}=[draw, signal, fill=gray, signal to=east, signal from=west, inner sep=1pt, minimum height=6pt]
\tikzstyle{lmat}=[draw, signal, fill=gray, signal to=west, signal from=east, inner sep=1pt, minimum height=6pt]
\tikzstyle{umat}=[draw, signal, fill=gray, signal to=north, signal from=south, inner sep=1pt, minimum width=6pt]
\tikzstyle{dmat}=[draw, signal, fill=gray, signal to=south, signal from=north, inner sep=1pt, minimum width=6pt]

\tikzstyle{simple}=[-]
\tikzstyle{hadamard edge}=[-, dashed, dash pattern=on 2pt off 0.5pt, thick, draw={rgb,255: red,68; green,136; blue,255}]
\tikzstyle{box edge}=[-, dashed, dash pattern=on 2pt off 0.5pt, thick, draw={rgb,255: red,203; green,192; blue,225}]
\tikzstyle{brace edge}=[-, tikzit draw=blue, decorate, decoration={brace,amplitude=1mm,raise=-1mm}]
\tikzstyle{diredge}=[->, thick]
\tikzstyle{double edge}=[-, double, shorten <=-1mm, shorten >=-1mm, double distance=2pt]
\tikzstyle{gray edge}=[-, {gray!60!white}]
\tikzstyle{pointer edge}=[->, very thick, gray]
\tikzstyle{boldedge}=[-, line width=1.0pt, shorten <=-0.17mm, shorten >=-0.17mm]
\tikzstyle{bidir edge}=[<->, very thick, draw={rgb,255: red,191; green,191; blue,191}]
\tikzstyle{purple edge}=[->, thick, draw={rgb,255: red,225; green,117; blue,216}]
\tikzstyle{green edge}=[->, thick, draw={rgb,255: red,167; green,231; blue,137}]
\tikzstyle{orange edge}=[->, thick, draw={rgb,255: red,245; green,170; blue,63}]
\tikzstyle{blue edge}=[->, thick, draw={rgb,255: red,68; green,136; blue,255}]
\tikzstyle{any edge}=[->, thick, draw=cyan]
\tikzstyle{red edge}=[->, thick, draw={rgb,255: red,255; green,136; blue,136}]
\tikzstyle{bidiredge}=[<->, thick]
\tikzstyle{dashed diredge}=[->, dashed, dash pattern=on 1pt off 0.5pt]
\tikzstyle{bidashed diredge}=[<->, dashed, dash pattern=on 1pt off 0.5pt]
\tikzstyle{gray fill}=[-, fill={rgb,255: red,234; green,234; blue,234}, draw=black]
\tikzstyle{white fill}=[-, fill=white]
\tikzstyle{arrow}=[->]
\tikzstyle{very thick}=[-, line width=1pt, tikzit draw=red]
\tikzstyle{pointille}=[dashed, -]
\tikzstyle{red}=[-, draw=red]
\tikzstyle{blue}=[-, draw=blue]
\tikzstyle{green}=[-, draw=green]
\tikzstyle{arrow}=[->]
\tikzstyle{strike}=[-, tikzit draw={rgb,255: red,191; green,0; blue,64}, strike through]
\tikzstyle{strike'}=[-, tikzit draw=cyan, strike bend]
\tikzstyle{dashed no arrow}=[-, dashed, dash pattern=on 1pt off 0.5pt]

\input{quantum.tikzdefs}

\AtEndDocument{\refstepcounter{theorem}\label{finalthm}}

\begin{document}

\title{Unifying Graph Measures and Stabilizer Decompositions \\ for the Classical Simulation of Quantum Circuits}

\author{Julien Codsi}%
\thanks{Supported by 
NSF Grant DMS-2348219, 
AFOSR grant  FA9550-22-1-0083  and  the Fonds de recherche du Qu\'{e}bec Grant 321124.}
\email{jc3530@princeton.edu}
\affiliation{Princeton University}

\author{Tuomas Laakkonen}
\thanks{TL acknowledges support from the US Department of Energy under Award No. DE-SC0020264}
\email{tsrl@mit.edu}
\affiliation{Massachusetts Institute of Technology}

 
\begin{abstract}
Various algorithms have been developed to simulate quantum circuits on classical hardware. Among the most prominent are approaches based on \emph{stabilizer decompositions} and \emph{tensor network contraction}. In this work, we present a unified framework that bridges these two approaches, placing them under a common formalism. Using this, we present two new algorithms to simulate an $n$-qubit circuit $C$: one that runs in $\tilde{O}(T^{\mathsf{tw}(C)})$ time and the other in $\tilde{O}(T^{\gamma\cdot \mathsf{tw}(C)})$ time, where $\mathsf{tw}(C)$ and $\mathsf{rw}(C)$ refer to the the tree-width and rank-width, respectively, of a tensor network associated to $C$, $T$ is the number of non-Clifford gates in $C$, and $\gamma \approx 3.42$. The proposed algorithms are simple, only require a linear amount of memory, are trivially parallelizable, and interact nicely with ZX-diagram simplification routines. Furthermore, we introduce the refined complexity measures \emph{focused tree-width} and \emph{focused rank-width}, which are always at least as efficient as their standard equivalent; these can be directly applied within our simulation algorithms, allowing for a more precise upper bound on the run time.
\end{abstract}


\maketitle

\setlength{\parskip}{1pt plus1pt minus1pt}
\setlength{\parindent}{15pt}

The classical simulation of quantum computing plays a critical role in the development, validation, and theoretical understanding of quantum algorithms. As we navigate the \emph{noisy intermediate-scale quantum} era, classical simulators serve an essential role in validating quantum hardware. Without the ability to model quantum behavior on classical hardware, we cannot easily verify the outputs of quantum processors or refine the algorithms intended to run on them. Despite the widely-held belief that simulating generic quantum systems classically requires time exponential in the system size, many quantum computations can be simulated with remarkable efficiency. A striking example of this is the celebrated \emph{Gottesman--Knill theorem} \cite{gottesman_heisenberg_1998}, which demonstrates that circuits composed entirely of Clifford gates are efficiently classically simulable. In contrast, several classes of algorithms are known that establish asymptotic bounds on the time complexity of general simulation in terms of different parameters. Two of the most prominent frameworks include \emph{stabilizer decomposition}, which scales with the number of non-Clifford gates (e.g., $T$-gates) required for a given computation, and \emph{tensor network methods}, which leverage graph-theoretic measures of the underlying circuit topology to optimize the computation.

Our results aim to unite these families of bounds under a single framework. In particular, in Section \ref{sec:rank-width} we give a strong simulation algorithm with a time complexity $\tilde{O}(T^{\gamma \cdot \rw(C)})$, where $T$ is the number of non-Clifford gates, $rw(C)$ is the \emph{rank-width} of an associated tensor network $C$, and $\gamma \approx 3.42$. Then in Section \ref{sec:tree-width}, we give an algorithm with time complexity complexity $\tilde{O}(T^{\tw(C)})$ where $\tw$ is the \emph{tree-width} of $C$. Our algorithms are trivially parallelizable, only use a linear amount of memory with respect to the number of gates, and can also be combined with simplification routines based on the ZX-calculus \cite{CD1,CD2}; this provides a speedup in practice while maintaining the runtime bounds based on graph-theoretic measures. Indeed, rank-width cannot increase while performing local complementation and pivoting, which are some of the main tools used in many simplification routines (see, for example, \cite{cliffsimp}). In Section \ref{sec:focused}, we introduce refined complexity measures called \emph{focused tree-width} and \emph{focused rank-width}. These are always at least as efficient as their standard equivalent and can be directly applied within our simulation algorithm, which allows us to more precisely upper bound the run time of the algorithms in Section \ref{sec:rank-width}. Finally, in Section~\ref{sec:benchmarks}, we empirically evaluate the performance of one of our algorithms across several families of random circuits. We observe that it performs particularly well when the edge density is very high, in contrast to typical stabilizer methods and approaches based on tree-width.

We remark that Kuyanov and Kissinger recently and independently exhibited an algorithm with qualitatively similar running times (only scaling exponentially with respect to rank-width) by a different method \cite{personal}.

The methodology developed in this work is highly versatile. Indeed, with straightforward modifications to our approach, one can derive analogous results for \emph{carving-width} and \emph{branch-width}. These encompass the primary findings of \cite{jakes2019carving} and \cite{ogorman_parameterization_2019} when specialized to the context of quantum circuits. We omit these derivations for brevity.

Our algorithm was found through the lens of ZX-calculus. As it will be used throughout this paper, we give a brief overview of the ZX-calculus~\cite{CD1,CD2} in the appendix. For an in-depth reference see~\cite{vandewetering2020zxcalculus}.
We address the case of \emph{strong} simulation, i.e.~calculating amplitudes of the circuit. This can be used to derive a \emph{weak} simulation algorithm, i.e.~sampling with a multiplicative linear overhead in the number of gates, see \cite{bravyi2021simulate}. Without loss of generality, we can assume that we want to know the amplitude $\bra{0^n}C\ket{0^n}$ for a circuit $C$.

\subsection*{Related Work}

The Gottesman-Knill theorem \cite{gottesman_heisenberg_1998} established the baseline for efficient simulation by restricting the gate set to the Clifford group. Subsequent work (for example \cite{bravyi_simulation_2019,bravyi_improved_2016}) extended this by developing algorithms that scale exponentially only with the $T$-count, utilizing the concept of stabilizer rank. A particularly useful tool with this approach is the ZX-calculus, which has been used to create several simulation algorithms based on stabilizer decompositions \cite{kissinger_simulating_2022,kissinger2022classical,sutcliffe_procedurally_2024,wille_basis_2022,codsi_cutting-edge_nodate,codsi_classically_2025,koch_speedy_2024}.

Parallel to these developments, \emph{tensor network methods} have proven highly effective for circuits with limited entanglement. Methods utilizing measures of graph width, such as \emph{tree-width} \cite{markov2008simulating}, and \emph{branch-width} \cite{ogorman_parameterization_2019} have been proposed to bound the cost of contracting these networks. In the case of measurement-based quantum computing, it was shown that it is possible to simulate a computation with time scaling only exponentially with the \emph{rank-width} (or Schmidt rank) of the cluster state \cite{van_den_nest_classical_2007}. 
However, the graph parameter-based and stabiliser-based approaches have largely remained distinct (with some exceptions \cite{codsi_classically_2025, kuyanov_rank-widthbased_nodate}). Our approach builds upon the graph-cutting framework for stabilizer decomposition, a technique initially introduced by \cite{codsi_cutting-edge_nodate} and subsequently refined in \cite{sutcliffe2024smarter, sutcliffe_procedurally_2024}.

Our work bridges this gap by demonstrating that the complexity of simulating a quantum circuit can be made to depend only on the rank-width of the tensor network associated to the circuit, and the number of non-Clifford operations. This follows the direction of recent research into the ZX-calculus and other graphical simplification rules \cite{kissinger_reducing_2020, fischbach2025exhaustivesearchquantumcircuit}, which seek to reduce the number of non-Clifford operations. By introducing the \emph{focused rank-width}, we provide a more granular complexity measure that accounts for the specific distribution of non-Clifford tensors within the graph, offering a more precise tool for determining the limits of classical simulation than standard graph width measures alone.


\section{Prerequisites}

\subsection{Graph-Theoretic Prerequisites}
\label{sec:graph-prereqs}

Let us formally define both \emph{tree-width} and \emph{rank-with} of a graph.

\begin{definition}
For a graph $G$, a \emph{tree decomposition} $(T, \chi)$ of $G$ consists of a tree $T$ and a map $\chi\colon V(T) \to 2^{V(G)}$ with the following properties:
\begin{enumerate}
\itemsep -.2em
    \item For every $v \in V(G)$, there exists $t \in V(T)$ such that $v \in \chi(t)$.

    \item For every $v_1v_2 \in E(G)$, there exists $t \in V(T)$ such that $v_1, v_2 \in \chi(t)$.

    \item For every $v \in V(G)$, the subgraph of $T$ induced by $\{t \in V(T) \mid v \in \chi(t)\}$ is connected.
\end{enumerate}
For each $t\in V(T)$, we refer to $\chi(t)$ as a \emph{bag of} $(T, \chi)$.  The \emph{width} of a tree decomposition $(T, \chi)$, denoted by $\mathrm{width}(T, \chi)$, is $\max_{t \in V(T)} |\chi(t)|-1$. The \emph{tree-width} of $G$, denoted by $\tw(G)$, is the minimum width of a tree decomposition of $G$.
\end{definition}

\begin{definition}

Let \( G = (V, E) \) be a graph.
For any subset \( X \subseteq V \), define the \emph{cut-rank} function
\[
\rho_G(X) = \mathrm{rank}_{\mathbb{F}_2}(M_{X, V \setminus X})
\]
where \( M_{X, V \setminus X} \) is the \( |X| \times |V \setminus X| \) binary biadjacency matrix over the field \( \mathbb{F}_2 \) (i.e \( M_{x,y} = 1 \) iff \( (x, y) \in E \) with \( x \in X\), \(y \notin X\)) and the rank is taken over the finite field \( \mathbb{F}_2 \). Note that $\rho_G(X) \leq \min\{|X|, |V\setminus X|\}$.
    
\end{definition}

\begin{definition}
A \emph{rank-decomposition} of \( G \) is a pair \( (T, \mu) \) where \( T \) is a subcubic tree (each internal node has degree three), and \( \mu: V \to \mathrm{leaves}(T) \) is a bijection mapping each vertex of \( G \) to a leaf of \( T \). For each edge \( e \in E(T) \) we define a partition \( (A_e, B_e) \) of \( V \) by deleting \( e \), and letting \( A_e \) and \( B_e \) be the sets of vertices mapped to the leaves of the resulting subtrees. The \emph{width} of the decomposition is $\max_{e \in E(T)} \rho_G(A_e)$.

\end{definition}
\begin{definition}
The \emph{rank-width} of \( G \), denoted \( \rw(G) \), is the minimum width over all possible rank-decompositions.
\end{definition}

Even though these parameters are defined on graphs, we extend them straightforwardly to ZX-diagrams by looking at the underlying graph, removing phases and colors. We can see there is a simple upper bound for the rank-width:

\begin{lemma}
    \label{thm:rwn3ub}
    For a graph $G$ on $n$ vertices, we have $\rw(G) \leq \lceil \frac{n}{3}\rceil$.
\end{lemma}

\begin{proof}
    Divide the set of vertices into thirds and construct an arbitrary rank-decomposition of each. Connect these decompositions with a single vertex of degree three after subdividing an edge from each of the resulting sub-cubic trees (to ensure that the resulting tree is sub-cubic). For any edge $e$, either $A_e$ or $B_e$ must have at most $\lceil\frac{n}{3}\rceil$ vertices. Therefore, $\rho_G(A_e) \leq \lceil\frac{n}{3}\rceil$, and hence $\rw(G) \leq \lceil\frac{n}{3}\rceil$.
\end{proof}

The next two lemmas show that all graphs with low rank-width or tree-width can be partitioned in a balanced way so that the associated cut is `small' in a certain sense; this will be essential for our algorithms. These are defined in terms of weight functions on the vertices:

\begin{definition}
    A \emph{weight function} on the vertices of a graph $G$ is a map $w : V(G) \to \mathbb{R}_{\geq 0}$. We extend this to any $X \subseteq V(G)$ as $w(X) = \sum_{v \in X} w(v)$. We will call $w(G) = w(V(G))$ the \emph{weight of} $G$.
\end{definition}

\begin{lemma}[{\cite{cygan2015parameterized}, Lemma 7.19}]\label{lemma: small separator}
    If a graph $G$ has $\tw(G)=k$, then for any weight function, there exists a separator set $X\subseteq V(G)$, with $|X|\leq k$, such that no component of $G\setminus X$ has more than half the weight of $G$.
\end{lemma}

\begin{lemma}\label{lemmma: low-rank sep}
    If a graph $G=(V,E)$ has $\rw(G)=k$, then for any weight function, there exists a partition $(X,V\setminus X)$ of $V$ such that $\rho_G(X) \leq k$ and such that both $X$ and $V\setminus X$ have at most $2/3$ the weight of $G$.
\end{lemma}
\begin{proof}
    Without loss of generality, let $w(G)=1$, and assume that no such partition exists. Let $(T,\mu)$ be the minimum-width rank-decomposition. We define an orientation of $T$ where every edge $e$ points toward the heaviest of $A_e$ and $B_e$. Every edge must be oriented towards a subtree with weight more than $2/3$, otherwise the partition $(A_e, B_e)$ would violate our assumption. Since every oriented tree has a vertex with no outgoing edges, let $v$ be such a vertex. $v$ cannot be a leaf so let $e_1,e_2,e_3$ be its adjacent edges, and $A_{e_i}$ be the side of the partition induced by $e_i$ not containing $v$. Since $A_{e_1},A_{e_2},A_{e_3}$ forms a partition of $V$, there is an $e_i$, say $e_1$, with $w(A_{e_1})\geq 1/3$. But since $e_1$ is oriented towards $v$, we also have that $w(A_{e_1}) < 1/3$, which is a contradiction.
\end{proof}

Note that the proof of Lemma \ref{lemmma: low-rank sep} is constructive; it shows that, given a minimal rank-decomposition $(T, \mu)$, there must exist an edge $e \in E(T)$ for which $(A_e, B_e)$ satisfies the conditions of the lemma. Such a partition can therefore be found easily by picking $e$ that minimizes $|w(A_e) - w(B_e)| = |2w(A_e) - w(G)|$. Moreover, this can be done in linear time: \cite[Section 5.3]{nouwt2022simulated} shows how to find an ordering of $V(G)$ for which every $A_e$ is contiguous; the prefix sum of the weight function with respect to this ordering can be computed to determine $w(A_e)$ for every edge $e$. Similarly, for Lemma \ref{lemma: small separator}, one of the bags of the minimum-width tree decomposition is a valid separator \cite{cygan2015parameterized}. 

We note that the relationship between both width parameters and balanced separators holds in both directions, up to constant factors. In particular, the existence of small balanced separators for every weight function implies bounded tree-width \cite{GraphMinorsII}. Similarly, the existence of balanced separators of low rank implies bounded rank-width \cite{rankwidthbalancedseparations}.
We complete this section with some linear algebra notation that will be useful.

\begin{definition}
    Let \( G = (V, E) \) be a graph and let $X,\subseteq V$. Let $A,B$ be subsets of $X$ and $V\setminus X$ respectively. We define the \emph{complete bipartite submatrix} $M_{X, V \setminus X}(A,B)$ as the matrix where \( (M_{X, V \setminus X}(A,B))_{x,y} = 1 \) iff \( x \in A \) and \( y \in B\).
\end{definition}

\begin{definition}
    Given a graph $G = (V, E)$ and a subset $X \subseteq V$, a \emph{bipartite decomposition} of size $k$ for the biadjacency matrix $M_{X, V \setminus X}$ consists of $k$ pairs of subsets $\{(A_1, B_1), \dots, (A_k, B_k)\}$, with each $A_i \subseteq X$ and $B_i \subseteq V\setminus X$, such that
    $$ M_{X,V\setminus X} = \sum_{i=1}^k M_{X, V \setminus X}(A_i,B_i) $$
    where the sum is taken over $\mathbb{F}_2$. 
\end{definition}

\begin{lemma}
    For every $G = (V, E)$ and $X \subseteq V$, $M_{X, V \setminus X}$ has a bipartite decomposition of size $\rho_G(X)$.
\end{lemma}
\begin{proof}
    This follows from the definition of rank over $\mathbb{F}_2$, and it can be constructed via Gaussian elimination \cite{Piziak1999}.
\end{proof}

\subsection{ZX-Calculus Prerequisites}
\label{sec:zx-prereqs}

The following observation will be useful.  By the definition of a green vertex, i.e.~\emph{Z-spiders}, as a linear map~\eqref{eq:def-Z-spider} implies that we can decompose it as a sum of diagrams containing \emph{X-spiders} (red vertices) via~\eqref{eq:state-ZX}:
\begin{equation}\label{eq:vertex-cut}
    \begin{tikzpicture}
	\begin{pgfonlayer}{nodelayer}
		\node [style=Z phase dot] (0) at (0, 0) {$\alpha$};
		\node [style=none] (3) at (1.25, 1) {};
		\node [style=none] (4) at (1.25, -1) {};
		\node [style=none] (6) at (1, 0.25) {$\vdots$};
		\node [style=none] (7) at (2.25, 0) {$=$};
		\node [style=none] (11) at (5.5, 1) {};
		\node [style=none] (12) at (5.5, -1) {};
		\node [style=none] (14) at (4.75, 0.25) {$\vdots$};
		\node [style=X dot] (16) at (4.25, 1) {};
		\node [style=X dot] (18) at (4.25, -1) {};
		\node [style=none] (19) at (6.25, 0) {$+$};
		\node [style=none] (22) at (10.25, 1) {};
		\node [style=none] (23) at (10.25, -1) {};
		\node [style=none] (25) at (10.25, 0.25) {$\vdots$};
		\node [style=X phase dot] (27) at (9, 1) {$\pi$};
		\node [style=X phase dot] (29) at (9, -1) {$\pi$};
		\node [style=none] (30) at (7.75, 0) {$\frac{e^{i\alpha}}{\sqrt{2}^k}$};
		\node [style=none] (31) at (3.25, 0) {};
		\node [style=none] (32) at (3.375, 0) {$\frac{1}{\sqrt{2}^k}$};
	\end{pgfonlayer}
	\begin{pgfonlayer}{edgelayer}
		\draw [in=180, out=75] (0) to (3.center);
		\draw [in=-180, out=-75] (0) to (4.center);
		\draw (16) to (11.center);
		\draw (18) to (12.center);
		\draw (27) to (22.center);
		\draw (29) to (23.center);
	\end{pgfonlayer}
\end{tikzpicture} \quad \forall \alpha \in [0,2\pi]
\end{equation}
We call this equation the \emph{vertex cut operation}, because, combined with simplifications, it can be used to remove vertices from a diagram. We will be interested in a particular form of ZX-diagrams:
\begin{definition}
A \emph{graph-like ZX-diagram} is a diagram where all spiders are green, and all edges except output and input wires are H-edges.
\end{definition}
The first step of our algorithms will be to transform our diagram into this form. It is possible to transform any ZX-diagram into a Graph-like form without increasing the tree-width of the underlying diagram. One only needs to perform the colour change rule on every red spider and get rid of every non-H-edge by using spider fusion. The underlying graph is only modified through edge contractions under which the tree-width is monotone \cite{GraphMinorsII}. It is unclear what effect this transformation has on rank-width, as it is not well-behaved under this operation, but some control is maintained, as rank-width is always upper-bounded by tree-width \cite{oum_rank-width_2008}. For graph-like ZX-diagrams, the two following equations \cite{pivotingZX} hold:
\begin{equation}\label{eq: local complementation}
    \begin{tikzpicture}
	\begin{pgfonlayer}{nodelayer}
		\node [style=Z phase dot] (0) at (-5, 0) {$\ \pm\frac\pi2\ $};
		\node [style=Z phase dot] (2) at (-7.5, -1.25) {$\alpha_1$};
		\node [style=Z phase dot] (4) at (-2.25, -1.25) {$\alpha_n$};
		\node [style=none] (5) at (-8.5, -2.25) {};
		\node [style=none] (6) at (-7.5, -2.25) {};
		\node [style=none] (7) at (-2.25, -2.25) {};
		\node [style=none] (8) at (-1.25, -2.25) {};
		\node [style=none] (9) at (-5, -2) {...};
		\node [style=none] (10) at (-8, -2.25) {...};
		\node [style=none] (11) at (-1.75, -2.25) {...};
		\node [style=none] (16) at (0, -1.5) {$\approx$};
		\node [style=none] (17) at (9.5, -1.25) {};
		\node [style=none] (18) at (1.25, -1.25) {};
		\node [style=none] (19) at (9, -1.25) {...};
		\node [style=none] (20) at (8.5, -1.25) {};
		\node [style=Z phase dot] (24) at (2.25, -0.25) {$\ \alpha_1\!\mp\!\frac\pi2\ $};
		\node [style=none] (26) at (2.25, -1.25) {};
		\node [style=none] (27) at (1.75, -1.25) {...};
		\node [style=Z phase dot] (30) at (8.5, -0.25) {$\ \alpha_n \!\mp\! \frac\pi2\ $};
		\node [style=none] (31) at (-6.5, -4.25) {};
		\node [style=Z phase dot] (32) at (-6.5, -3.25) {$\alpha_2$};
		\node [style=none] (33) at (-7, -4.25) {...};
		\node [style=none] (34) at (-7.5, -4.25) {};
		\node [style=Z phase dot] (35) at (-3.25, -3.25) {$\ \alpha_{n\!-\!1}\ $};
		\node [style=none] (36) at (-3.25, -4.25) {};
		\node [style=none] (37) at (-2.25, -4.25) {};
		\node [style=none] (38) at (-2.75, -4.25) {...};
		\node [style=Z phase dot] (39) at (3.75, -2.75) {$\ \alpha_2\!\mp\!\frac\pi2\ $};
		\node [style=none] (40) at (3.25, -3.75) {...};
		\node [style=none] (41) at (2.75, -3.75) {};
		\node [style=none] (42) at (3.75, -3.75) {};
		\node [style=Z phase dot] (43) at (6.75, -2.75) {$\ \alpha_{n\!-\!1} \!\mp\! \frac\pi2\ $};
		\node [style=none] (44) at (7.75, -3.75) {};
		\node [style=none] (45) at (7.25, -3.75) {...};
		\node [style=none] (46) at (6.75, -3.75) {};
		\node [style=none] (47) at (5.25, -1) {...};
		\node [style=none] (48) at (0, -1) {(LC)};
	\end{pgfonlayer}
	\begin{pgfonlayer}{edgelayer}
		\draw (5.center) to (2);
		\draw (6.center) to (2);
		\draw (7.center) to (4);
		\draw (8.center) to (4);
		\draw [style=hadamard edge] (2) to (0);
		\draw [style=hadamard edge] (4) to (0);
		\draw (18.center) to (24);
		\draw (26.center) to (24);
		\draw (20.center) to (30);
		\draw (17.center) to (30);
		\draw (34.center) to (32);
		\draw (31.center) to (32);
		\draw (36.center) to (35);
		\draw (37.center) to (35);
		\draw (41.center) to (39);
		\draw (42.center) to (39);
		\draw (46.center) to (43);
		\draw (44.center) to (43);
		\draw [style=hadamard edge] (24) to (30);
		\draw [style=hadamard edge] (30) to (43);
		\draw [style=hadamard edge] (43) to (39);
		\draw [style=hadamard edge] (39) to (24);
		\draw [style=hadamard edge] (24) to (43);
		\draw [style=hadamard edge] (39) to (30);
		\draw [style=hadamard edge] (0) to (32);
		\draw [style=hadamard edge] (0) to (35);
	\end{pgfonlayer}
\end{tikzpicture}
\end{equation}
and
\begin{equation}\label{eq: pivot}
    \begin{tikzpicture}
	\begin{pgfonlayer}{nodelayer}
		\node [style=Z phase dot] (0) at (-8.5, 2.25) {$j\pi$};
		\node [style=Z phase dot] (2) at (-10.25, 1.75) {$\alpha_1$};
		\node [style=none] (16) at (0, 0) {$\approx$};
		\node [style=Z phase dot] (32) at (-10.25, 0.25) {$\alpha_n$};
		\node [style=Z phase dot] (37) at (-6.5, -1.25) {$\beta_m$};
		\node [style=Z phase dot] (41) at (-6.5, 0.25) {$\beta_1$};
		\node [style=Z phase dot] (45) at (-3, 1.75) {$\gamma_1$};
		\node [style=Z phase dot] (49) at (-3, 0.25) {$\gamma_l$};
		\node [style=Z phase dot] (51) at (-4.5, 2.25) {$k\pi$};
		\node [style=none] (52) at (-10.25, 1) {...};
		\node [style=none] (53) at (-6.5, -0.5) {...};
		\node [style=none] (54) at (-3, 1) {...};
		\node [style=Z phase dot] (55) at (2.75, 0.75) {$\ \alpha_n+k\pi\ $};
		\node [style=Z phase dot] (56) at (6.5, -1.25) {$\ \beta_1+(j+k+1)\pi\ $};
		\node [style=none] (57) at (6.5, -2) {...};
		\node [style=Z phase dot] (58) at (6.5, -2.75) {$\ \beta_m+(j+k+1)\pi\ $};
		\node [style=Z phase dot] (59) at (10.25, 2.25) {$\ \gamma_1+j\pi\ $};
		\node [style=Z phase dot] (60) at (2.75, 2.25) {$\ \alpha_1+k\pi\ $};
		\node [style=none] (61) at (10.25, 1.5) {...};
		\node [style=none] (62) at (2.75, 1.5) {...};
		\node [style=Z phase dot] (63) at (10.25, 0.75) {$\ \gamma_l+j\pi\ $};
		\node [style=none] (64) at (-9.75, 2.5) {};
		\node [style=none] (65) at (-10.75, 2.5) {};
		\node [style=none] (66) at (-10.25, 2.5) {...};
		\node [style=none] (67) at (-10.75, -0.5) {};
		\node [style=none] (68) at (-10.25, -0.5) {...};
		\node [style=none] (69) at (-9.75, -0.5) {};
		\node [style=none] (70) at (-7, -2) {};
		\node [style=none] (71) at (-6.5, -2) {...};
		\node [style=none] (72) at (-6, -2) {};
		\node [style=none] (73) at (-7.25, 1.25) {};
		\node [style=none] (74) at (-6.5, 1) {...};
		\node [style=none] (75) at (-5.75, 1.25) {};
		\node [style=none] (76) at (-3.5, 2.5) {};
		\node [style=none] (77) at (-3, 2.5) {...};
		\node [style=none] (78) at (-2.5, 2.5) {};
		\node [style=none] (79) at (-3, -0.5) {...};
		\node [style=none] (80) at (-3.5, -0.5) {};
		\node [style=none] (81) at (-2.5, -0.5) {};
		\node [style=none] (82) at (2.25, 3) {};
		\node [style=none] (83) at (2.75, 3) {...};
		\node [style=none] (84) at (3.25, 3) {};
		\node [style=none] (85) at (2.75, 0) {...};
		\node [style=none] (86) at (2.25, 0) {};
		\node [style=none] (87) at (3.25, 0) {};
		\node [style=none] (88) at (6, -0.5) {};
		\node [style=none] (89) at (6.5, -0.5) {...};
		\node [style=none] (90) at (7, -0.5) {};
		\node [style=none] (91) at (6.5, -3.5) {...};
		\node [style=none] (92) at (6, -3.5) {};
		\node [style=none] (93) at (7, -3.5) {};
		\node [style=none] (94) at (9.75, 3) {};
		\node [style=none] (95) at (10.25, 3) {...};
		\node [style=none] (96) at (10.75, 3) {};
		\node [style=none] (97) at (10.25, 0) {...};
		\node [style=none] (98) at (9.75, 0) {};
		\node [style=none] (99) at (10.75, 0) {};
		\node [style=none] (100) at (0, 0.625) {(P)};
	\end{pgfonlayer}
	\begin{pgfonlayer}{edgelayer}
		\draw [style=hadamard edge] (2) to (0);
		\draw [style=hadamard edge] (0) to (32);
		\draw [style=hadamard edge] (0) to (51);
		\draw [style=hadamard edge, bend right] (0) to (37);
		\draw [style=hadamard edge, bend right] (0) to (41);
		\draw [style=hadamard edge, bend left] (51) to (37);
		\draw [style=hadamard edge, bend left, looseness=0.75] (51) to (41);
		\draw [style=hadamard edge] (51) to (45);
		\draw [style=hadamard edge] (51) to (49);
		\draw [style=hadamard edge] (60) to (56);
		\draw [style=hadamard edge] (60) to (58);
		\draw [style=hadamard edge] (55) to (56);
		\draw [style=hadamard edge] (55) to (58);
		\draw [style=hadamard edge] (59) to (56);
		\draw [style=hadamard edge] (59) to (58);
		\draw [style=hadamard edge] (56) to (63);
		\draw [style=hadamard edge] (63) to (58);
		\draw [style=hadamard edge] (60) to (59);
		\draw [style=hadamard edge] (60) to (63);
		\draw [style=hadamard edge] (55) to (63);
		\draw [style=hadamard edge] (55) to (59);
		\draw (64.center) to (2);
		\draw (2) to (65.center);
		\draw (67.center) to (32);
		\draw (69.center) to (32);
		\draw (70.center) to (37);
		\draw (72.center) to (37);
		\draw (41) to (75.center);
		\draw (41) to (73.center);
		\draw (45) to (78.center);
		\draw (45) to (76.center);
		\draw (49) to (80.center);
		\draw (49) to (81.center);
		\draw (60) to (82.center);
		\draw (60) to (84.center);
		\draw (59) to (94.center);
		\draw (59) to (96.center);
		\draw (98.center) to (63);
		\draw (99.center) to (63);
		\draw (86.center) to (55);
		\draw (87.center) to (55);
		\draw (56) to (88.center);
		\draw (56) to (90.center);
		\draw (92.center) to (58);
		\draw (93.center) to (58);
	\end{pgfonlayer}
\end{tikzpicture}
\end{equation}
These are called \emph{local complementation} and \emph{pivoting}.
These two operations have been used, very effectively, to simplify ZX-diagrams by removing Clifford spiders. In particular, they allow one to give a fully diagrammatic proof of the Gottesman-Knill theorem \cite{cliffsimp}. On graph-like ZX diagrams, these operations act on the underlying graph as \emph{vertex minors}, and do not increase the rank-width \cite{vertexminor}, although they may increase the tree-width. 

\section{An Algorithm Scaling with the Tree-Width}
\label{sec:tree-width}

We start by giving an algorithm that will strongly simulate a quantum circuit in a time that scales exponentially with the tree-width. Even though this is, more often than not, slower than what we will later show, it encompasses most of the ideas needed. 
Consider  \cref{algo: tree width}.

\begin{algorithm}[ht]
\begin{algorithmic}[1]
\If {$|D|=1$}
     \State \Return the value of $D$
\EndIf
\State  $D = D$ in graph-like form
\State $S:= $ balanced separator of $D$ with uniform weights such that $|S|\leq \tw(D)$
\State $acc:=0$
\For{$p_1,\dots,p_{|S|} =0,1$}
    \State $D':= D$ where each $s_i\in S$ is replaced by a red spider with phase $p_i\pi$ (applying \eqref{eq:vertex-cut})
    \State $D' = D'$ where all the red spiders are colour-changed to green and fused with their neighbours
    \State $acc = acc~+$ recurse on $D'$ 
\EndFor
\State \Return $acc$
\end{algorithmic}
\caption{Simulation scaling with the tree-width.}
\label{algo: tree width}

\end{algorithm}

The number of times this algorithm is called to solve an instance of size $n$ follows the recursion $T(n) \leq 2^{\tw(D)}T(n/2)$, from which we see $T(n) \leq n^{\tw(D)}$.
Although computing tree-width is NP-hard, it is fixed-parameter tractable, and can be solved in linear time for fixed parameters \cite{fptTreewidth}. In practice, there exist many efficient algorithms that give `good' tree decompositions (see, for example, \cite{gray2021hyper}).
Since the decomposition phase is computationally negligible compared to the main algorithm, we omit its cost from the leading term, resulting in a total runtime of $\tilde{O}(n^{\tw(D)})$. 
Here, the tree-width of the original diagram is used to upper-bound the branching factor, but as the diagrams get smaller, one would expect the tree-width to also decrease; this effect is not captured by the running time bound presented here.

Drawing on ideas from the diagramatic proof of the Gottesman-Knill theorem \cite{cliffsimp}, it is possible to significantly improve upon \cref{algo: tree width} with some minor changes. Let $NC(D)$ be the number of non-Clifford spiders in a ZX-diagram $D$ (that is, spiders with phases that are not multiples of $\pi/2$). Consider \cref{algo: tree width clifford}.

\begin{algorithm}[ht]
\begin{algorithmic}[1]
\If {$NC(D)=0$}
     \State \Return the value of $D$ using the Gottesman-Knill theorem
\EndIf
\State $D = D$ in graph-like form
\State $w :=$ the weight function with zero weight for Clifford spiders and uniform non-zero weight otherwise
\State $S := $ balanced separator of $D$, weighted by $w$, such that $|S|\leq \tw(D)$
\State $acc:=0$
\For{$p_1,\dots,p_{|S|} =0,1$}
    \State Let $D':= D$ where each $s_i\in S$ is replaced by a red spider with phase $p_i\pi$ (applying \eqref{eq:vertex-cut})
    \State $D' = D'$ where all the red spiders are colour-changed to green and fused with their neighbours
    \State $acc = acc~+ $ recurse on $D'$ 
\EndFor
\State \Return $acc$
\end{algorithmic}
\caption{Simulation scaling with the tree-width and the number of non-Clifford tensors.}
\label{algo: tree width clifford}

\end{algorithm}

Repeating the same analysis as above shows that the time complexity of the algorithm is
$\tilde{O}\big(NC(D)^{\tw(D)}\big)$. Therefore, this algorithm yields an alternative proof that quantum circuit simulation scales exponentially with the tree-width~\cite{treewidthContraction}. While the running time presented here is suboptimal compared to the $\tilde{O}(\exp(\tw(D)))$ obtained by~\cite{treewidthContraction}, our approach requires only linear memory, is highly parallelizable, and incorporates knowledge about the non-Clifford gates.

For fixed tree-width and a large number of non-Clifford gates, Algorithm \ref{algo: tree width clifford} compares favourably with stabilizer decomposition methods that do not (directly) exploit structural information about the underlying graph. Such approaches typically scale as $\tilde{O}(2^{\alpha\, NC(D)})$, for some constant $\alpha$ depending on the circuit class (see, for example,~\cite{kissinger2022classical}). However, below the threshold
$$NC(D) \approx \frac{\tw(D)}{\alpha}\log\left(\frac{\tw(D)}{\alpha}\right)$$
these other methods become more efficient.
This observation suggests that incorporating an early-stopping criterion may be beneficial. Instead of terminating the recursion only when $NC(D)=0$, one can stop at the above threshold and switch to a stabilizer decomposition algorithm. This hybrid strategy improves the polynomial factors in the running time. In practice, the optimal threshold should be tuned empirically, as polynomial overheads are non-negligible.

\section{An Algorithm Scaling with the Rank-Width}
\label{sec:rank-width}

Here, we will aim to use a very similar approach as with tree-width, namely, directly using efficient separators instead of the structure of the decomposition emerging from the control of the parameter. Just like for tree-width, where we could delete vertices of the separator one at a time using \eqref{eq:vertex-cut}, we will need an operation that will decrease the rank of a cut while leaving the rest of the diagram unchanged at the cost of increasing the number of diagrams to consider.
The idea will be to introduce two adjacent spiders connected in an appropriate way to the rest of the ZX-diagram and then apply \eqref{eq: pivot}. To make sure that the newly introduced spiders do not change the value of our diagram, we will make use of the $\pi$-copy rule. Let $A,B$ be subsets of spiders in which we would like to add $H$-edges between every pair of spiders. Then we can write this as a sum of diagrams:
\begin{equation}\label{eq:bipartite sum}
    \begin{tikzpicture}
	\begin{pgfonlayer}{nodelayer}
		\node [style=Z dot] (0) at (2.5, 9.5) {};
		\node [style=none] (16) at (-2, -4) {$\approx$};
		\node [style=Z dot] (51) at (5, 9.5) {};
		\node [style=none] (100) at (-2, -3.375) {(P)};
		\node [style=Z phase dot] (102) at (-10.25, 9) {$\alpha_1$};
		\node [style=Z phase dot] (103) at (-10.25, 7) {$\alpha_n$};
		\node [style=none] (107) at (-10.25, 8.25) {$\vdots$};
		\node [style=none] (109) at (-11.5, 9.75) {};
		\node [style=none] (110) at (-11.5, 8.5) {};
		\node [style=none] (111) at (-11.5, 9.25) {$\vdots$};
		\node [style=none] (112) at (-11.5, 7.5) {};
		\node [style=none] (113) at (-11.5, 7) {$\vdots$};
		\node [style=none] (114) at (-11.5, 6.25) {};
		\node [style=none] (121) at (-2, 8) {$\approx$};
		\node [style=none] (122) at (-2, 8.625) {($\pi$-copy)};
		\node [style=X dot] (123) at (2.5, 10.75) {};
		\node [style=X dot] (124) at (5, 10.75) {};
		\node [style=none] (145) at (-2, 1.75) {$=$};
		\node [style=none] (146) at (-2, 2.375) {(SF)};
		\node [style=none] (149) at (-14, 2) {$\sum\limits_{j,k\in\{0,1\}}$};
		\node [style=none] (150) at (-16.25, 2) {$\approx$};
		\node [style=none] (171) at (-11, 10.5) {};
		\node [style=none] (172) at (-9.5, 10.5) {};
		\node [style=none] (173) at (-10.3, 11) {$A$};
		\node [style=none] (174) at (-5.5, 10.5) {};
		\node [style=none] (175) at (-4.25, 10.5) {};
		\node [style=none] (176) at (-4.875, 11) {$B$};
		\node [style=none] (178) at (0.25, -3.75) {$\sum\limits_{j,k\in\{0,1\}}$};
		\node [style=Z phase dot] (179) at (-5, 9) {$\gamma_1$};
		\node [style=Z phase dot] (180) at (-5, 7) {$\gamma_l$};
		\node [style=none] (181) at (-5, 8.25) {$\vdots$};
		\node [style=none] (182) at (-3.75, 9.75) {};
		\node [style=none] (183) at (-3.75, 8.5) {};
		\node [style=none] (184) at (-3.75, 9.25) {$\vdots$};
		\node [style=none] (185) at (-3.75, 7.5) {};
		\node [style=none] (186) at (-3.75, 7) {$\vdots$};
		\node [style=none] (187) at (-3.75, 6.25) {};
		\node [style=Z phase dot] (188) at (1, 9) {$\alpha_1$};
		\node [style=Z phase dot] (189) at (1, 7) {$\alpha_n$};
		\node [style=none] (190) at (1, 8.25) {$\vdots$};
		\node [style=none] (191) at (-0.25, 9.75) {};
		\node [style=none] (192) at (-0.25, 8.5) {};
		\node [style=none] (193) at (-0.25, 9.25) {$\vdots$};
		\node [style=none] (194) at (-0.25, 7.5) {};
		\node [style=none] (195) at (-0.25, 7) {$\vdots$};
		\node [style=none] (196) at (-0.25, 6.25) {};
		\node [style=Z phase dot] (203) at (6.25, 9) {$\gamma_1$};
		\node [style=Z phase dot] (204) at (6.25, 7) {$\gamma_l$};
		\node [style=none] (205) at (6.25, 8.25) {$\vdots$};
		\node [style=none] (206) at (7.5, 9.75) {};
		\node [style=none] (207) at (7.5, 8.5) {};
		\node [style=none] (208) at (7.5, 9.25) {$\vdots$};
		\node [style=none] (209) at (7.5, 7.5) {};
		\node [style=none] (210) at (7.5, 7) {$\vdots$};
		\node [style=none] (211) at (7.5, 6.25) {};
		\node [style=Z dot] (212) at (-9.25, 3.25) {};
		\node [style=Z dot] (213) at (-6.75, 3.25) {};
		\node [style=Z phase dot] (214) at (-9.25, 4.5) {$j\pi$};
		\node [style=Z phase dot] (215) at (-6.75, 4.5) {$k\pi$};
		\node [style=Z phase dot] (216) at (-10.75, 2.75) {$\alpha_1$};
		\node [style=Z phase dot] (217) at (-10.75, 0.75) {$\alpha_n$};
		\node [style=none] (218) at (-10.75, 2) {$\vdots$};
		\node [style=none] (219) at (-12, 3.5) {};
		\node [style=none] (220) at (-12, 2.25) {};
		\node [style=none] (221) at (-12, 3) {$\vdots$};
		\node [style=none] (222) at (-12, 1.25) {};
		\node [style=none] (223) at (-12, 0.75) {$\vdots$};
		\node [style=none] (224) at (-12, 0) {};
		\node [style=Z phase dot] (225) at (-5.5, 2.75) {$\gamma_1$};
		\node [style=Z phase dot] (226) at (-5.5, 0.75) {$\gamma_l$};
		\node [style=none] (227) at (-5.5, 2) {$\vdots$};
		\node [style=none] (228) at (-4.25, 3.5) {};
		\node [style=none] (229) at (-4.25, 2.25) {};
		\node [style=none] (230) at (-4.25, 3) {$\vdots$};
		\node [style=none] (231) at (-4.25, 1.25) {};
		\node [style=none] (232) at (-4.25, 0.75) {$\vdots$};
		\node [style=none] (233) at (-4.25, 0) {};
		\node [style=none] (234) at (0.25, 2) {$\sum\limits_{j,k\in\{0,1\}}$};
		\node [style=Z phase dot] (238) at (5, 3.25) {$j\pi$};
		\node [style=Z phase dot] (239) at (7.5, 3.25) {$k\pi$};
		\node [style=Z phase dot] (240) at (3.5, 2.75) {$\alpha_1$};
		\node [style=Z phase dot] (241) at (3.5, 0.75) {$\alpha_n$};
		\node [style=none] (242) at (3.5, 2) {$\vdots$};
		\node [style=none] (243) at (2.25, 3.5) {};
		\node [style=none] (244) at (2.25, 2.25) {};
		\node [style=none] (245) at (2.25, 3) {$\vdots$};
		\node [style=none] (246) at (2.25, 1.25) {};
		\node [style=none] (247) at (2.25, 0.75) {$\vdots$};
		\node [style=none] (248) at (2.25, 0) {};
		\node [style=Z phase dot] (249) at (8.75, 2.75) {$\gamma_1$};
		\node [style=Z phase dot] (250) at (8.75, 0.75) {$\gamma_l$};
		\node [style=none] (251) at (8.75, 2) {$\vdots$};
		\node [style=none] (252) at (10, 3.5) {};
		\node [style=none] (253) at (10, 2.25) {};
		\node [style=none] (254) at (10, 3) {$\vdots$};
		\node [style=none] (255) at (10, 1.25) {};
		\node [style=none] (256) at (10, 0.75) {$\vdots$};
		\node [style=none] (257) at (10, 0) {};
		\node [style=Z phase dot] (260) at (4, -2.75) {$\ \alpha_1+k\pi\ $};
		\node [style=Z phase dot] (261) at (4, -4.75) {$\ \alpha_n+k\pi\ $};
		\node [style=none] (262) at (4, -3.5) {$\vdots$};
		\node [style=none] (263) at (2.25, -2) {};
		\node [style=none] (264) at (2.25, -3.25) {};
		\node [style=none] (265) at (2.25, -2.5) {$\vdots$};
		\node [style=none] (266) at (2.25, -4.25) {};
		\node [style=none] (267) at (2.25, -4.75) {$\vdots$};
		\node [style=none] (268) at (2.25, -5.5) {};
		\node [style=Z phase dot] (269) at (8.25, -2.75) {$\ \gamma_1+j\pi\ $};
		\node [style=Z phase dot] (270) at (8.25, -4.75) {$\ \gamma_l+j\pi\ $};
		\node [style=none] (271) at (8.25, -3.5) {$\vdots$};
		\node [style=none] (272) at (10, -2) {};
		\node [style=none] (273) at (10, -3.25) {};
		\node [style=none] (274) at (10, -2.5) {$\vdots$};
		\node [style=none] (275) at (10, -4.25) {};
		\node [style=none] (276) at (10, -4.75) {$\vdots$};
		\node [style=none] (277) at (10, -5.5) {};
	\end{pgfonlayer}
	\begin{pgfonlayer}{edgelayer}
		\draw [style=hadamard edge] (0) to (51);
		\draw (109.center) to (102);
		\draw (102) to (110.center);
		\draw (112.center) to (103);
		\draw (114.center) to (103);
		\draw [style=simple] (123) to (0);
		\draw [style=simple] (124) to (51);
		\draw [style=brace edge] (171.center) to (172.center);
		\draw [style=brace edge] (174.center) to (175.center);
		\draw (182.center) to (179);
		\draw (179) to (183.center);
		\draw (185.center) to (180);
		\draw (187.center) to (180);
		\draw (191.center) to (188);
		\draw (188) to (192.center);
		\draw (194.center) to (189);
		\draw (196.center) to (189);
		\draw (206.center) to (203);
		\draw (203) to (207.center);
		\draw (209.center) to (204);
		\draw (211.center) to (204);
		\draw [style=hadamard edge] (188) to (0);
		\draw [style=hadamard edge] (0) to (189);
		\draw [style=hadamard edge] (51) to (204);
		\draw [style=hadamard edge] (51) to (203);
		\draw [style=hadamard edge] (212) to (213);
		\draw [style=simple] (214) to (212);
		\draw [style=simple] (215) to (213);
		\draw (219.center) to (216);
		\draw (216) to (220.center);
		\draw (222.center) to (217);
		\draw (224.center) to (217);
		\draw (228.center) to (225);
		\draw (225) to (229.center);
		\draw (231.center) to (226);
		\draw (233.center) to (226);
		\draw [style=hadamard edge] (216) to (212);
		\draw [style=hadamard edge] (212) to (217);
		\draw [style=hadamard edge] (213) to (226);
		\draw [style=hadamard edge] (213) to (225);
		\draw (243.center) to (240);
		\draw (240) to (244.center);
		\draw (246.center) to (241);
		\draw (248.center) to (241);
		\draw (252.center) to (249);
		\draw (249) to (253.center);
		\draw (255.center) to (250);
		\draw (257.center) to (250);
		\draw [style=simple, loop] (238) to ();
		\draw [style=hadamard edge] (240) to (238);
		\draw [style=hadamard edge] (238) to (241);
		\draw [style=simple, loop] (239) to ();
		\draw [style=hadamard edge] (239) to (250);
		\draw [style=hadamard edge] (239) to (249);
		\draw [style=hadamard edge] (238) to (239);
		\draw (263.center) to (260);
		\draw (260) to (264.center);
		\draw (266.center) to (261);
		\draw (268.center) to (261);
		\draw (272.center) to (269);
		\draw (269) to (273.center);
		\draw (275.center) to (270);
		\draw (277.center) to (270);
		\draw [style=hadamard edge] (260) to (269);
		\draw [style=hadamard edge] (269) to (261);
		\draw [style=hadamard edge] (261) to (270);
		\draw [style=hadamard edge] (270) to (260);
	\end{pgfonlayer}
\end{tikzpicture}
\end{equation}
Since, by the Hopf rule (see \cite[Section 9.2]{CKbook}), two green spiders sharing two $H$-edges is equivalent to them not sharing any edges, this indeed adds the $H$-edges (modulo $2$) of a complete bipartite graph between $A$ and $B$ at the cost of creating $4$ diagrams. We call this operation a \emph{complete bipartite sum}, and using this new operation, we can now make a simulation algorithm that scales with the rank-width.

\begin{algorithm}[ht]
\caption{Simulation scaling with the rank-width.}
\label{algo: rank-width}
\begin{algorithmic}[1]
\If {$|D|=1$}
     \State \Return the value of $D$
\EndIf
\State  $D = D$ in graph-like form
\State $X,\bar{X}$:= a $2/3$-balanced partition of $D$ with uniform weights such that $\rho_D(X)\leq \rw(D)$
\State $\{(A_1,B_1), \dots (A_{\rho_D(X)}, B_{\rho_D(X)})\}$ := the minimal bipartite decomposition of $M_{X,\bar{X}}$
\State $acc:=0$
\State $D':= D$ where all the $H$-edges between $X$ and $\bar{X}$ are removed
\For{$k_1,\dots,k_{\rho_D(X)} =0,1$ and $j_i,\dots,j_{\rho_D(X)} =0,1$}
    \State $D'':=D'$
    \For{$i = 1,\dots,\rho_D(X)$}
        \State $D''= D''$ where we add $k_i\pi$ to the phase of all the spiders in $A_i$
        \State $D''= D''$ where we add $j_i\pi$ to the phase of all the spiders in $B_i$
    \EndFor
    \State $acc = acc ~+$ recurse on $D''$ 
\EndFor
\State \Return $acc$
\end{algorithmic}
\end{algorithm}

The number of times this algorithm is called to solve a given instance of size $n$ follows the recursion $T(n) \leq 4^{\rw(D)}T(2n/3)$, from which we find $T(n) \leq O(4^{\tw(D)\log_{3/2}n}) = O(n^{\log_{3/2}(4)\rw(D)})$ (note that $\log_{3/2}(4) = \gamma \approx 3.42$). Similarly to tree-width, computing rank-width is NP-hard \cite{oum2008approximating}, but for every $k$, there exists a polynomial algorithm to find a rank-decomposition with width at most $k$ if it exists \cite{fptRankwidth}. Moreover, there exist heuristics to efficiently find relatively good decomposition (see, for example \cite{heuristicRw,annealingRw}). Therefore, we can compute a rank-decomposition once and then find any balanced separator in linear time. The time to compute the rank-decomposition is thus negligible compared to the rest of the algorithm. Here again, we can improve upon this algorithm by utilizing the Gottesman–Knill algorithm for Clifford diagrams, which leads to \cref{algo: rank-width non clifford}.

\begin{algorithm}[ht]
\caption{Simulation scaling with the rank-width and the number of non-Clifford tensors.}
\label{algo: rank-width non clifford}
\begin{algorithmic}[1]
\If {$NC(D)=1$}
     \State \Return the value of $D$ using the Gottesman-Knill theorem
\EndIf
\State  $D = D$ in graph-like form
\State $w :=$ the weight function with weight zero for Clifford spiders and uniform non-zero weight otherwise
\State $X,\bar{X}$:= a $2/3$-balanced partition of $D$, weighted by $w$, such that $\rho_D(X)\leq \rw(D)$
\State $\{(A_1,B_1), \dots (A_{\rho_D(X)}, B_{\rho_D(X)})\}$ := the minimal bipartite decomposition of $M_{X,\bar{X}}$
\State $acc:=0$
\State $D':= D$ where all the $H$-edges between $X$ and $\bar{X}$ are removed
\For{$k_1,\dots,k_{\rho_D(X)} =0,1$ and $j_i,\dots,j_{\rho_D(X)} =0,1$}
    \State $D'':=D'$
    \For{$i = 1,\dots,\rho_D(X)$}
        \State $D''= D''$ where we add $k_i\pi$ to the phase of all the spiders in $A_i$
        \State $D''= D''$ where we add $j_i\pi$ to the phase of all the spiders in $B_i$
    \EndFor
    \State $acc = acc ~+$ recurse on $D''$ 
\EndFor
\State \Return $acc$
\end{algorithmic}
\end{algorithm}

By the same analysis, we find that this algorithm runs in time $\tilde{O}(NC(D)^{\log_{3/2}(4)\cdot \rw(D)})$. Here again, it would be beneficial to optimize the threshold to switch to a stabilizer decomposition algorithm. Because the states utilized in measurement-based quantum computing (MBQC) are inherently graph-like, this algorithm provides an alternative proof that any universal resource for MBQC must possess unbounded rank-width, which is the central result of \cite{mbqcUnbounded}. Furthermore, while we achieve a bound that is suboptimal in running time to the $\tilde{O}(2^{rw(D)})$ obtained by \cite{van_den_nest_classical_2007}, that method is limited to MBQC, while ours is not. Our approach additionally requires only linear memory, is highly parallelizable, and integrates knowledge regarding non-Clifford gates.

\subsection{Mitigating the Factor of Four for Bipartite Sums}
\label{sec:mitigating}

The complete bipartite sum operation \eqref{eq:bipartite sum} is a very similar idea to the \emph{subgraph complement} operation introduced in \cite{codsi_cutting-edge_nodate}, and it is possible to make a less efficient version of complete bipartite sums using subgraph complement. Indeed, since a pivot can be decomposed into $3$ consecutive local complementations, we could have implemented this operation this way, but with $8$ diagrams instead of $4$.

\begin{question}
    Is it possible to realize a complete bipartite sum in fewer than $4$ diagrams?
\end{question}

On the other hand, the vertex deletion operation \eqref{eq:vertex-cut} creates only $2$ new diagrams. Since this factor appears in the basis of the exponent of the running time, this seemingly small difference can have a huge impact, which might undermine the advantage of utilizing \cref{algo: rank-width non clifford} compared to \cref{algo: tree width clifford}, even though rank-width is always bounded by tree-width (see \cite{oum_rank-width_2008}). To mitigate this effect, we combine the use of \eqref{eq:vertex-cut} and \eqref{eq:bipartite sum}. Let $D$ be a diagram with underlying graph $G=(V,E)$ and $X$. It is possible to achieve the separation $(X,V \sm X)$ by deleting the vertices of either $N(X)$ or $N(V\sm X)$ using the more efficient vertex deletions. In general, we use the following

\begin{definition}
    \label{def:deletionbipartitedecomposition}
    Let $k_1,k_2,k_3\in\naturel$, \( G = (V, E) \) be a graph and $X\subseteq V$ and let $Y=V\sm X$. Then, we say that $M_{X,Y}$ has a \emph{mixed deletion-bipartite decomposition} if there exists $(A_1,B_1),\dots, (A_{k_1},B_{k_1}) \subseteq 2^{X}\times 2^{Y}$, $u_1,\dots,u_{k_2}\in X$ and  $v_1,\dots,v_{k_3}\in Y$ such that 
    \begin{equation}
        M_{X,Y} = \sum_{i=1}^{k_1} M_{X,Y}(A_i,B_i) +   \sum_{i=1}^{k_2} M_{X, Y}(\set{u_i},N(u_i)\cap Y ) +  \sum_{i=1}^{k_3} M_{X, Y}(N(v_i)\cap X, \set{v_i})
    \end{equation}
    where the sum is taken over $\mathbb{F}_2$. The score of such a decomposition is $2k_1 +k_2+k_3$.
\end{definition}
Terms in the second and third sums represent vertex deletion, while the first accounts for bipartite sums, which are twice as costly
(alternatively, consider starting with $M_{X, Y}$ and performing either row deletion, column deletion, or adding sub-matrices full of ones). We define the minimal score for a given $X$:

\begin{definition}
    \label{def:mixed-cut-rank}
    Given $G = (V, E)$ and $X \subseteq V$, the \emph{mixed cut-rank} of $X$ is the minimum score of any mixed deletion-bipartite decomposition of $M_{X,V\setminus X}$.
\end{definition}

It is not obvious \emph{a priori} how to compute the mixed cut-rank. However, given a partition $(X,Y)$, it is easy to compute a mixed deletion-bipartite decomposition with a `good' score. First, note that vertex deletions commute with the other operations, hence they can be performed upfront. Furthermore, each term in the sum can decrease the rank of $M_{X,Y}$ by at most one; by the definition of rank, there is always a way to achieve this reduction via complete bipartite sums. This suggests a greedy approach: delete rows or columns that reduce the rank until none remain, then transition to bipartite sums. In general, this is not optimal.
If we restrict the decomposition such that $k_1 + k_2 + k_3 = \operatorname{rank}(M_{X,Y})$, the optimal score can, in fact, be computed in polynomial time by relating it to the maximum independent sets of a particular bipartite graph; we omit the details here because it is complicated and not usually helpful in practice.

\begin{question}
    Is it possible to efficiently compute the mixed cut-rank?
\end{question}

We will incorporate this into our algorithms by defining a \emph{mixed rank-width}. Algorithms \ref{algo: rank-width} and \ref{algo: rank-width non clifford} can be adapted directly by replacing the bipartite decomposition with the mixed decomposition. 

\begin{definition}
    The \emph{mixed rank-width} of a rank-decomposition is the maximum mixed cut-rank of any partition implied by an edge of the decomposition. The \emph{mixed rank-width} of a graph is the minimal mixed rank-width of any rank-decomposition.
\end{definition}

The $\lceil \frac{n}{3}\rceil$ upper bound for rank-width extends to mixed rank-width, even though the mixed cut-rank can, in general, be twice as large as the usual cut-rank.

\begin{lemma}
    The mixed rank-width of a graph on $n$ vertices is at most $\lceil \frac{n}{3}\rceil$.
\end{lemma}

\begin{proof}
    Let $G$ be a graph of $n$ vertices. Construct the rank-decomposition $(T, \mu)$ of $G$ described in Lemma \ref{thm:rwn3ub}. By definition, for every edge $e \in T$, we have $\min \{|A_e|, |B_e|\} \leq \lceil\frac{n}{3}\rceil$. Hence, for every edge there is a mixed bipartite-deletion decomposition with score at most $\lceil\frac{n}{3}\rceil$ that consists of deleting every vertex. Therefore, the mixed cut-rank of every edge is at most $\lceil\frac{n}{3}\rceil$, and then so is the mixed rank-width.
\end{proof}

\section{Specialized graph parameters and ZX simplifications}
\label{sec:focused}

\subsection{New Parameters}
In the previous section, our best upper-bound on the size of a cut separating non-Clifford gates originated from the tree-width and the rank-width of the diagram. Both widths are measures that minimize the size of a large number of cuts, not only the one that we end up using. This makes it so that our bounds based on those measures might be very far from the actual time complexity of our algorithms. For a more nuanced analysis, we introduce refined graph parameters: the \emph{focused tree-width} and \emph{focused rank-width}.

\begin{definition}
    Let $G$ be a graph and $S\subseteq V(G)$.
    A \emph{focused rank-decomposition} of \( G \) with special set $S$ is a pair \( (T, \mu) \) where \( T \) is a subcubic tree (each internal node has degree 3), and \( \mu:  \mathrm{leaves}(T) \to 2^V \) is a partition of $V$ mapping each leaf to a subset of $V(G)$ containing at most one vertex of $S$.

    As with a usual rank-decomposition, each edge \( e \in E(T) \) defines a partition \( (A_e, B_e) \) of \( V \) by deleting \( e \), and letting \( A_e \) and \( B_e \) be the sets of vertices mapped to a leaf of the resulting subtrees.
    The \emph{width} of the decomposition is $\max_{e \in E(T)} \rho_G(A_e)$.
\end{definition}

\begin{definition}
Let $G$ be a graph and let $S\subseteq V(G)$.
The \emph{focused rank-width} of \( G \) with special set $S$, denoted \( \frw{S}(G) \), is the minimum width over all possible focused rank-decompositions: $\frw{S}(G) = \min_{(T, \mu)} \max_{e \in E(T)} \rho_G(A_e)$.
\end{definition}

\begin{lemma}\label{lem: bound on frw}
    Let $G$ be a graph and let $S\subseteq V(G)$, then $\frw{S}(G) < |S|$.
\end{lemma}
\begin{proof}
    Let $s_1,\dots,s_{|S|}$ be the vertices of $S$. Let $T$ be any subcubic tree with $|S|$ leaves which we will denote as $l_1,,l_{|S|}$. Let $\mu$ be defined by
    \[\mu(l_i)= \begin{cases}
       \{l_1\} \cup V\setminus S  &\quad\text{if } i=1\\
       \{i\} &\quad\text{otherwise} \\
     \end{cases}\]
    For any edge of $T$, one of the two parts of the induced partition has at most $|S|-1$ vertices. Since the rank of a matrix is bounded by the minimum of its dimensions, the result follows.
\end{proof}

\begin{lemma}\label{lemmma: low focused rank sep}
    Let $G$ be a graph and let $S\subseteq V(G)$.
    If a graph $G=(V,E)$ has $\frw{S}(G)=k$, then for any weight function $w$ on $S$, there exists a partition $(X,V\setminus X)$ of $V$ such that $\rho_G(X) \leq k$ and such that both $X$ and $V\setminus X$ have at most $2/3$ the weight of $G$.
\end{lemma}
\begin{proof}
    Proceed exactly as in the proof of Lemma~\ref{lemmma: low-rank sep}.
\end{proof}

\begin{theorem}\label{thm: frw sim}
    Let $C$ be a graph-like ZX-diagram and let $S$ be its non-Clifford vertices.
    There is an algorithm to classically simulate $C$ that runs in time $\tilde{O}\left(|S|^{\log_{3/2}(4)\cdot \frw{S}(C)}+f(S,C)\right)$ where $f$ is the time taken to compute the focused rank-width with special set $S$ on $C$.
\end{theorem}
\begin{proof}
    By noticing that in \cref{algo: rank-width non clifford}, the weight function is only non-zero on the non-Clifford vertices, we can apply Lemma~\ref{lemmma: low focused rank sep} to get this bound on the rank of the separators.
\end{proof}



It is more difficult to exploit this focused relaxation for tree decompositions. However, drawing from the connection between balanced separators and tree-width, we define the following:
\begin{definition}
    Let $G$ be a graph and $S\subseteq V(G)$. The \emph{focused tree-width} of $G$ with special set $S$, denoted $\ftw{S}(G)$, is the minimum $k$ for which the following holds: for every weight function $w$ for which $w(V(G)\setminus S)=0$, there exists $X\subseteq V(G)$ with $|X|\leq k$ such that no connected component of $G\setminus X$ has more than half the weight of $G$.
\end{definition}
The analysis of \cref{algo: tree width clifford} extends straightforwardly with this measure.
We defer the study of the computability of focused widths to further work.

\subsection{ZX-calculus Simplifications}
There has been a lot of effort in the ZX-calculus literature to design simplification and compilation routines for ZX-diagrams. Those have been used effectively to achieve efficient classical simulation algorithms by combining them with stabilizer decomposition in a `decompose-simplify-repeat' paradigm (e.g see \cite{kissinger_simulating_2022}). An interesting possibility for these approaches is to combine them with tensor network contraction algorithms (such as \cite{Gray2021hyperoptimized}). This kind of hybrid contraction and stabiliser decomposition algorithm has been mentioned in \cite{codsi_cutting-edge_nodate}. However, one of the main drawbacks of combining these kinds of strategies with ZX simplification routines is that we usually lose all control over some of the original graph parameters of the tensor network, as noted in \cite{codsi_cutting-edge_nodate, sutcliffe_procedurally_2024}. For instance, the pivot rule, which is extensively used in simplification routines, can increase the tree-width of a diagram without bound. This makes any interesting theoretical bounds very hard to achieve and, in practice, brings little to no advantage. While recent efforts to predict structural changes during post-decomposition simplification have yielded significant practical improvements \cite{sutcliffe_procedurally_2024}, these methods have yet to maintain rigorous control over graph parameters, such as the existence of good cuts.

This issue can mostly be alleviated by using (focused) rank-width instead of tree-width. Indeed, (focused) rank-width of a graph can only decrease under \emph{vertex-minors} \cite{vertexminor}. This implies that both local complementation and pivoting do not increase (focused) rank-width. Similarly, all but one of the ZX-calculus rules can only decrease the (focused) rank-width. The only exception is the spider-fusion rule, which can increase (focused) rank-width by at most one. On the other hand, using a tree-width-based algorithm, one can use spider-fusion, as tree-width is monotone under contraction \cite{GraphMinorsII}. 

\section{Benchmarks}
\label{sec:benchmarks}

\subsection{Implementation}

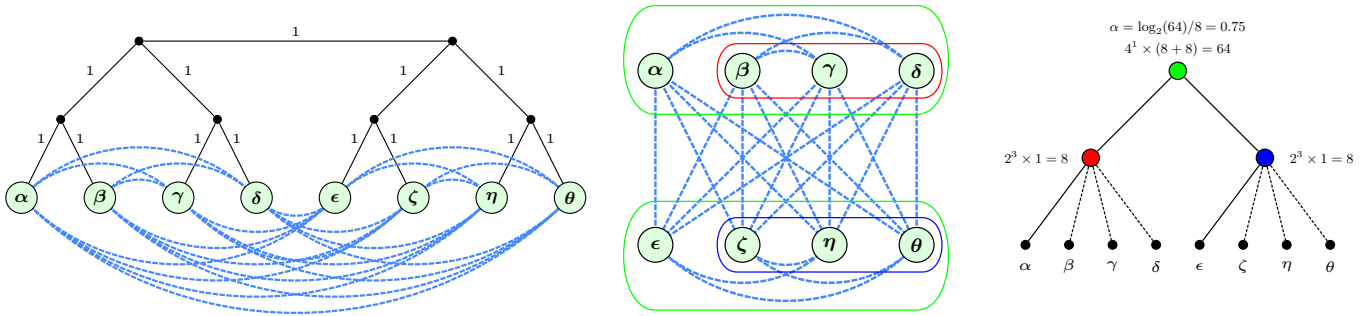
\begin{figure}
    \resizebox{\linewidth}{!}{\scalebox{0.9}{$\begin{tikzpicture}
	\begin{pgfonlayer}{nodelayer}
		\node [style=Z phase dot] (0) at (-1, -1) {$\delta$};
		\node [style=Z phase dot] (1) at (-5, -1) {$\beta$};
		\node [style=Z phase dot] (2) at (1, -1) {$\epsilon$};
		\node [style=Z phase dot] (3) at (3, -1) {$\zeta$};
		\node [style=Z phase dot] (4) at (5, -1) {$\eta$};
		\node [style=Z phase dot] (5) at (7, -1) {$\theta$};
		\node [style=Z phase dot] (6) at (-3, -1) {$\gamma$};
		\node [style=Z phase dot] (7) at (-7, -1) {$\alpha$};
		\node [style=vertex] (8) at (6, 1) {};
		\node [style=vertex] (9) at (2, 1) {};
		\node [style=vertex] (10) at (4, 3) {};
		\node [style=vertex] (11) at (-2, 1) {};
		\node [style=vertex] (13) at (-6, 1) {};
		\node [style=vertex] (14) at (-4, 3) {};
		\node [style=none] (15) at (-6.5, 0.5) {\tiny $1$};
		\node [style=none] (16) at (-5.5, 0.5) {\tiny $1$};
		\node [style=none] (17) at (-2.5, 0.5) {\tiny $1$};
		\node [style=none] (18) at (-1.5, 0.5) {\tiny $1$};
		\node [style=none] (19) at (1.5, 0.5) {\tiny $1$};
		\node [style=none] (20) at (2.5, 0.5) {\tiny $1$};
		\node [style=none] (21) at (5.5, 0.5) {\tiny $1$};
		\node [style=none] (22) at (6.5, 0.5) {\tiny $1$};
		\node [style=none] (23) at (0, 3.25) {\tiny $1$};
		\node [style=none] (24) at (5.25, 2.25) {\tiny $1$};
		\node [style=none] (25) at (2.75, 2.25) {\tiny $1$};
		\node [style=none] (26) at (-2.75, 2.25) {\tiny $1$};
		\node [style=none] (27) at (-5.25, 2.25) {\tiny $1$};
	\end{pgfonlayer}
	\begin{pgfonlayer}{edgelayer}
		\draw (8) to (4);
		\draw (8) to (5);
		\draw (9) to (2);
		\draw (9) to (3);
		\draw (10) to (9);
		\draw (10) to (8);
		\draw (11) to (6);
		\draw (11) to (0);
		\draw (13) to (7);
		\draw (13) to (1);
		\draw [style=hadamard edge, bend right=45] (0) to (2);
		\draw [style=hadamard edge, bend right=45] (0) to (3);
		\draw [style=hadamard edge, bend right=45] (0) to (4);
		\draw [style=hadamard edge, bend right=45] (0) to (5);
		\draw [style=hadamard edge, bend right=45] (6) to (2);
		\draw [style=hadamard edge, bend right=45, looseness=1.25] (6) to (3);
		\draw [style=hadamard edge, bend right=45] (6) to (4);
		\draw [style=hadamard edge, bend right=45] (6) to (5);
		\draw [style=hadamard edge, bend right=45] (1) to (2);
		\draw [style=hadamard edge, bend right=45] (1) to (3);
		\draw [style=hadamard edge, bend right=45] (1) to (4);
		\draw [style=hadamard edge, bend right=45] (1) to (5);
		\draw [style=hadamard edge, bend right=45] (7) to (2);
		\draw [style=hadamard edge, bend right=45] (7) to (3);
		\draw [style=hadamard edge, bend right=45] (7) to (4);
		\draw [style=hadamard edge, bend right=45] (7) to (5);
		\draw [style=hadamard edge, bend left=315] (6) to (7);
		\draw [style=hadamard edge, bend left=315] (0) to (7);
		\draw [style=hadamard edge, bend left=315] (0) to (1);
		\draw [style=hadamard edge, bend left=315] (6) to (1);
		\draw [style=hadamard edge, bend right=45] (4) to (3);
		\draw [style=hadamard edge, bend left=315] (4) to (2);
		\draw [style=hadamard edge, bend left=45] (2) to (5);
		\draw [style=hadamard edge, bend left=45] (3) to (5);
		\draw (14) to (11);
		\draw (14) to (13);
		\draw (10) to (14);
	\end{pgfonlayer}
\end{tikzpicture}$}\hspace{3mm}$\begin{tikzpicture}
	\begin{pgfonlayer}{nodelayer}
		\node [style=Z phase dot] (0) at (3, 2) {$\delta$};
		\node [style=Z phase dot] (1) at (-1, 2) {$\beta$};
		\node [style=Z phase dot] (2) at (-3, -2) {$\epsilon$};
		\node [style=Z phase dot] (3) at (-1, -2) {$\zeta$};
		\node [style=Z phase dot] (4) at (1, -2) {$\eta$};
		\node [style=Z phase dot] (5) at (3, -2) {$\theta$};
		\node [style=Z phase dot] (6) at (1, 2) {$\gamma$};
		\node [style=Z phase dot] (7) at (-3, 2) {$\alpha$};
		\node [style=none] (8) at (-3, -1) {};
		\node [style=none] (9) at (-3, -3.5) {};
		\node [style=none] (10) at (3, -3.5) {};
		\node [style=none] (11) at (3, -1) {};
		\node [style=none] (12) at (-3, 3.5) {};
		\node [style=none] (13) at (-3, 1) {};
		\node [style=none] (14) at (3, 1) {};
		\node [style=none] (15) at (3, 3.5) {};
		\node [style=none] (16) at (-1.125, 2.625) {};
		\node [style=none] (17) at (-0.25, 2.5) {};
		\node [style=none] (18) at (3.125, 1.375) {};
		\node [style=none] (19) at (-1.125, 1.375) {};
		\node [style=none] (20) at (3.125, 2.625) {};
		\node [style=none] (21) at (-1.125, -1.375) {};
		\node [style=none] (22) at (3.125, -2.625) {};
		\node [style=none] (23) at (-1.125, -2.625) {};
		\node [style=none] (24) at (3.125, -1.375) {};
		\node [style=gn] (25) at (9, 2) {};
		\node [style=rn] (26) at (7, 0) {};
		\node [style=vertex] (27) at (5.5, -2) {};
		\node [style=vertex] (28) at (6.5, -2) {};
		\node [style=vertex] (29) at (7.5, -2) {};
		\node [style=vertex] (30) at (8.5, -2) {};
		\node [style=monbl] (31) at (11, 0) {};
		\node [style=vertex] (32) at (9.5, -2) {};
		\node [style=vertex] (33) at (10.5, -2) {};
		\node [style=vertex] (34) at (11.5, -2) {};
		\node [style=vertex] (35) at (12.5, -2) {};
		\node [style=none] (36) at (5.5, -2.5) {\tiny $\alpha$};
		\node [style=none] (37) at (6.5, -2.5) {\tiny $\beta$};
		\node [style=none] (38) at (7.5, -2.5) {\tiny $\gamma$};
		\node [style=none] (39) at (8.5, -2.5) {\tiny $\delta$};
		\node [style=none] (40) at (9.5, -2.5) {\tiny $\epsilon$};
		\node [style=none] (41) at (10.5, -2.5) {\tiny $\zeta$};
		\node [style=none] (42) at (11.5, -2.5) {\tiny $\eta$};
		\node [style=none] (43) at (12.5, -2.5) {\tiny $\theta$};
		\node [style=none] (44) at (9, 2.5) {\scalebox{0.5}{$4^1 \times (8 + 8) = 64$}};
		\node [style=none] (46) at (5.75, 0) {\scalebox{0.5}{$2^3 \times 1 = 8$}};
		\node [style=none] (47) at (12.25, 0) {\scalebox{0.5}{ $2^3 \times 1 = 8$}};
		\node [style=none] (48) at (9, 3) {\scalebox{0.5}{$\alpha = \log_2(64) / 8 = 0.75$}};
	\end{pgfonlayer}
	\begin{pgfonlayer}{edgelayer}
		\draw [style=hadamard edge] (0) to (2);
		\draw [style=hadamard edge] (0) to (3);
		\draw [style=hadamard edge] (0) to (4);
		\draw [style=hadamard edge] (0) to (5);
		\draw [style=hadamard edge] (6) to (2);
		\draw [style=hadamard edge] (6) to (3);
		\draw [style=hadamard edge] (6) to (4);
		\draw [style=hadamard edge] (6) to (5);
		\draw [style=hadamard edge] (1) to (2);
		\draw [style=hadamard edge] (1) to (3);
		\draw [style=hadamard edge] (1) to (4);
		\draw [style=hadamard edge] (1) to (5);
		\draw [style=hadamard edge] (7) to (2);
		\draw [style=hadamard edge] (7) to (3);
		\draw [style=hadamard edge] (7) to (4);
		\draw [style=hadamard edge] (7) to (5);
		\draw [style=hadamard edge, bend left=315] (6) to (7);
		\draw [style=hadamard edge, bend left=315] (0) to (7);
		\draw [style=hadamard edge, bend left=315] (0) to (1);
		\draw [style=hadamard edge, bend left=315] (6) to (1);
		\draw [style=hadamard edge, bend left=45] (4) to (3);
		\draw [style=hadamard edge, bend right=315] (4) to (2);
		\draw [style=hadamard edge, bend right=45] (2) to (5);
		\draw [style=hadamard edge, bend right=45] (3) to (5);
		\draw [style=green, bend left=90] (11.center) to (10.center);
		\draw [style=green] (10.center) to (9.center);
		\draw [style=green, bend right=270] (9.center) to (8.center);
		\draw [style=green] (8.center) to (11.center);
		\draw [style=green, bend left=90] (15.center) to (14.center);
		\draw [style=green] (14.center) to (13.center);
		\draw [style=green, bend right=270] (13.center) to (12.center);
		\draw [style=green] (12.center) to (15.center);
		\draw [style=red] (18.center) to (19.center);
		\draw [style=red, bend left=90, looseness=1.25] (19.center) to (16.center);
		\draw [style=red] (16.center) to (20.center);
		\draw [style=red, bend left=90, looseness=1.25] (20.center) to (18.center);
		\draw [style=blue] (22.center) to (23.center);
		\draw [style=blue, bend left=90, looseness=1.25] (23.center) to (21.center);
		\draw [style=blue, in=180, out=0] (21.center) to (24.center);
		\draw [style=blue, bend left=90, looseness=1.25] (24.center) to (22.center);
		\draw (25) to (26);
		\draw (25) to (31);
		\draw (26) to (27);
		\draw (31) to (32);
		\draw [style=dashed no arrow] (26) to (28);
		\draw [style=dashed no arrow] (26) to (29);
		\draw [style=dashed no arrow] (26) to (30);
		\draw [style=dashed no arrow] (31) to (33);
		\draw [style=dashed no arrow] (31) to (34);
		\draw [style=dashed no arrow] (31) to (35);
	\end{pgfonlayer}
\end{tikzpicture}$}
    \caption{An example of constructing a rank decomposition and tree of cuts for a ZX-diagram. The diagram we use in this example consists of eight nodes with arbitrary phases, connected as a clique with four edges removed. \emph{Left:} This graph has rank-width one, and mixed rank-width two. A rank decomposition of the graph is given by the black nodes and edges that connect the green Z-spiders of the ZX-diagram. Each edge is labelled with its cut-rank. \emph{Center:} Here we show a potential decomposition of the graph, first into the two green sub-diagrams obtained by performing a single bipartite decomposition cut along the center. Each of the sub-diagrams is then reduced to only one vertex by cutting the three vertices outlined in red and blue, respectively. \emph{Right:} The same decomposition of the graph organized as a tree of cuts. Each interior node represents one decomposition operation, and the leaves represent the vertices of the graph. Leaves connected to interior nodes by dashed edges are the vertices to be cut. Each interior node is labelled with the total number of terms represented by its subtree. This example decomposes $8$ vertices into $64$ terms, resulting in $\alpha = 0.75$.}
    \label{fig:cuttreeexample}
\end{figure}

We implemented Algorithm \ref{algo: rank-width non clifford} using the QuiZX library \cite{quizx} for working with ZX-diagrams. Our goal is to find a decomposition of a ZX-diagram using bipartite decompositions and vertex deletions that minimizes the number of terms in the resulting sum of diagrams. If a diagram of $n$ non-Clifford vertices can be decomposed into $k$ terms, we define
$$ \alpha = \frac{\log_2(k)}{n} $$
as a measure of the efficiency of the decomposition that can be compared independent of the size $n$. We therefore seek to minimize $\alpha$ -- this measure has been widely used to describe stabilizer decomposition methods for classical simulation of quantum circuits \cite{kissinger_simulating_2022}\cite{sutcliffe_procedurally_2024}.

For Algorithm \ref{algo: rank-width non clifford}, rank decompositions were computed using the simulated annealing method of \cite{nouwt2022simulated}. This was modified to replace the cut-rank, which is used to calculate the width of each edge, with the mixed cut-rank (Definition \ref{def:mixed-cut-rank}), approximated using the greedy method suggested in Section \ref{sec:mitigating}. For all experiments, we used $10^5$ annealing steps and kept all other parameters the same as in \cite{nouwt2022simulated}.
After obtaining a rank decomposition, we recursively construct the tree of cuts that will be performed on the diagram. In each step, we iterate over all the partitions induced by edges of the rank decomposition, and select the cut that has a low score $\alpha_\mathrm{eff}$, called \emph{effective} $\alpha$, which was introduced in \cite{codsi_cutting-edge_nodate}. This is defined as follows: suppose that the cut being considered splits a ZX-diagram with $n$ non-Clifford vertices into two parts, with $a$ and $b$ non-Clifford vertices, using $2^w$ terms. If similar cuts could be found recursively in each of the resulting sub-diagrams, the total time complexity would scale with the recurrence:
$$ T(n) = 2^w \left[T(a) + T(b)\right]$$
If we assume that $T(n)$ has the form $T(n) = 2^{\alpha_\mathrm{eff} n}$, then given $w$, $a$, and $b$, the above equation can be solved for $\alpha_\mathrm{eff}$. In particular, we solve it by fixed-point iteration, after rewriting it into the following form:
$$ \alpha_\mathrm{eff} = \frac{w}{n - \max\{a, b\}} + \frac{1}{n - \max\{a, b\}}\log_2\left(1 + 2^{-\alpha_\mathrm{eff} |a - b|}\right) $$
At each step we pick a cut at random with probability proportional to $(\alpha_\mathrm{eff}T)^{-1}$, where $T = 0.05$ is a temperature parameter. For each diagram, we construct $5000$ cut trees in this way and select the best. An example is shown in Figure \ref{fig:cuttreeexample}. From the cut tree, we can determine the exact value of $\alpha$ by walking the tree bottom-up to compute the total number of terms in the decomposition; hence, we do not actually need to perform the decomposition, which allows us to estimate costs much larger than could be run practically.
\subsection{Results}
\noindent We benchmark our method on three families of graph-like ZX-diagrams:
\begin{enumerate}
    \item Er\H{o}s-R\'enyi random graphs with between $16$ and $64$ vertices, and edge probabilities from $p = 0.1$ to $p = 1$. Each vertex is set to an arbitrary phase. The results are shown in Figure \ref{fig:erdosrenyi}.
    \item Uniformly random Clifford+$R_Z$ gate circuits. We used circuits of between $16$ and $40$ qubits and $100$ to $1000$ gates. Phase gates were selected with a fixed probability of $0.2$, and $H$, $S$, and $CNOT$ gates were weighted equally. The results are shown in Figure \ref{fig:cliffordt}.
    \item Random Clifford+$R_Z$ circuits obtained by concatenating random Pauli gadgets \cite{pauliGadget}, which are subcircuits that act on many qubits. Each Pauli gadget is chosen uniformly at random, subject to the constraint that its support is at most half the qubits. We take between $8$ and $32$ qubits, with the number of Pauli gadgets between one and four times the number of qubits. The results are shown in Figure \ref{fig:pauli}.
\end{enumerate}
In each instance, we applied the \texttt{full-reduce} procedure of \cite{cliffsimp} to translate and simplify the circuit into a graph-like ZX-diagram. Note that the values of $\alpha$ shown in all of these plots refer to the number of non-Clifford vertices \emph{after} simplification of the ZX-diagram, to better measure the properties of our method, rather than the efficacy of the simplification procedure.

The results of our benchmarks are discussed in the captions of the relevant figures.  Significantly, we observe that our algorithm performs well when the edge density is either very high or very low, in contrast to typical stabilizer methods and tree-width-based methods. However, we note here that, across all the benchmarks, the values of $\alpha$ that we observe are, more often than not, significantly higher than for stabilizer decomposition methods for simulating Clifford+T circuits (e.g \cite{kissinger_simulating_2022,codsi_cutting-edge_nodate}). While this indicates that our method is not generally competitive for Clifford+T circuits, it has the advantage that it depends only on the structure of the ZX-diagram, and not the phases attached to each vertex. Therefore, it can be used for circuits containing arbitrary phase gates, while most stabilizer decomposition methods cannot make progress in these cases. A notable exception is \cite{sutcliffe_procedurally_2024}; the algorithms presented here would benefit from a more detailed comparison with their method in the future. We also mention that the chosen families are, in some sense, the worst case for our algorithms. Indeed, random graphs typically maximize most graph widths (for example, $\tw(G(n,1/2)) = \rw(G(n,1/2)) = \Omega(n)$ almost surely \cite{treewidthRandom,rankwidthRandom}). One would expect better performance in more structured circuits such as when the topology is restricted.

\begin{figure}
    \centering
    \hspace*{-\linewidth}\includegraphics[width=\linewidth]{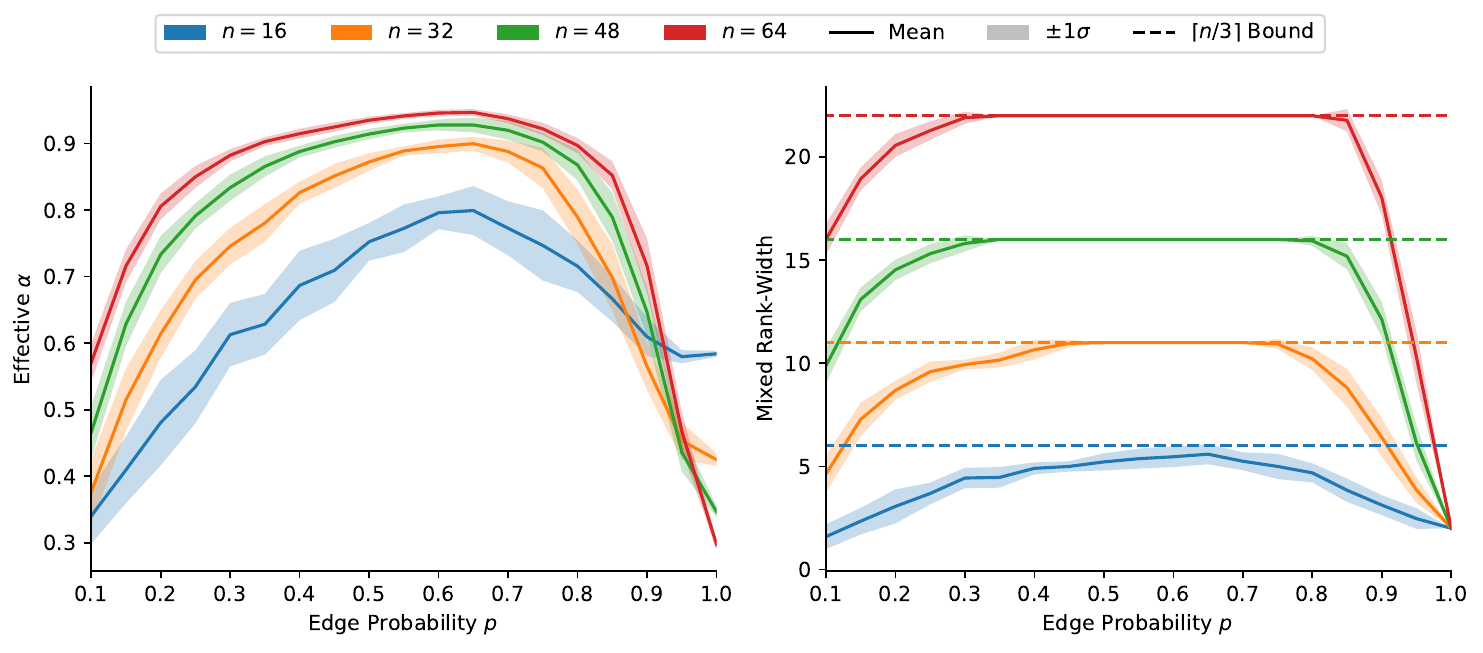}
    \caption{The mixed rank-width and computed $\alpha$ values to decompose an Erd\H{o}s-R\'enyi random graph with a varying number of vertices $n$ and probability of each edge $p$. We evaluated graphs with $n = 16$ to $64$ and $p = 0.1$ to $1$. We do not evaluate $p < 0.1$ to ensure that the graphs are not comprised of many small disconnected components. We see that when $p$ is small or large, the mixed rank-width and $\alpha$ decrease, while for intermediate $p$ the graphs almost always have maximal mixed rank-width. This is to be expected as rank-width is maximal with high probability for any constant $p$ as $n \to \infty$ \cite{rwRandom}. Note that the maximum $\alpha$ is not at $p = \frac{1}{2}$ because at lower densities vertex deletions are more prevalent, while at higher densities bipartite decompositions are required.}
    \label{fig:erdosrenyi}
\end{figure}

\begin{figure}
    \centering
    \hspace*{-\linewidth}\includegraphics[width=\linewidth]{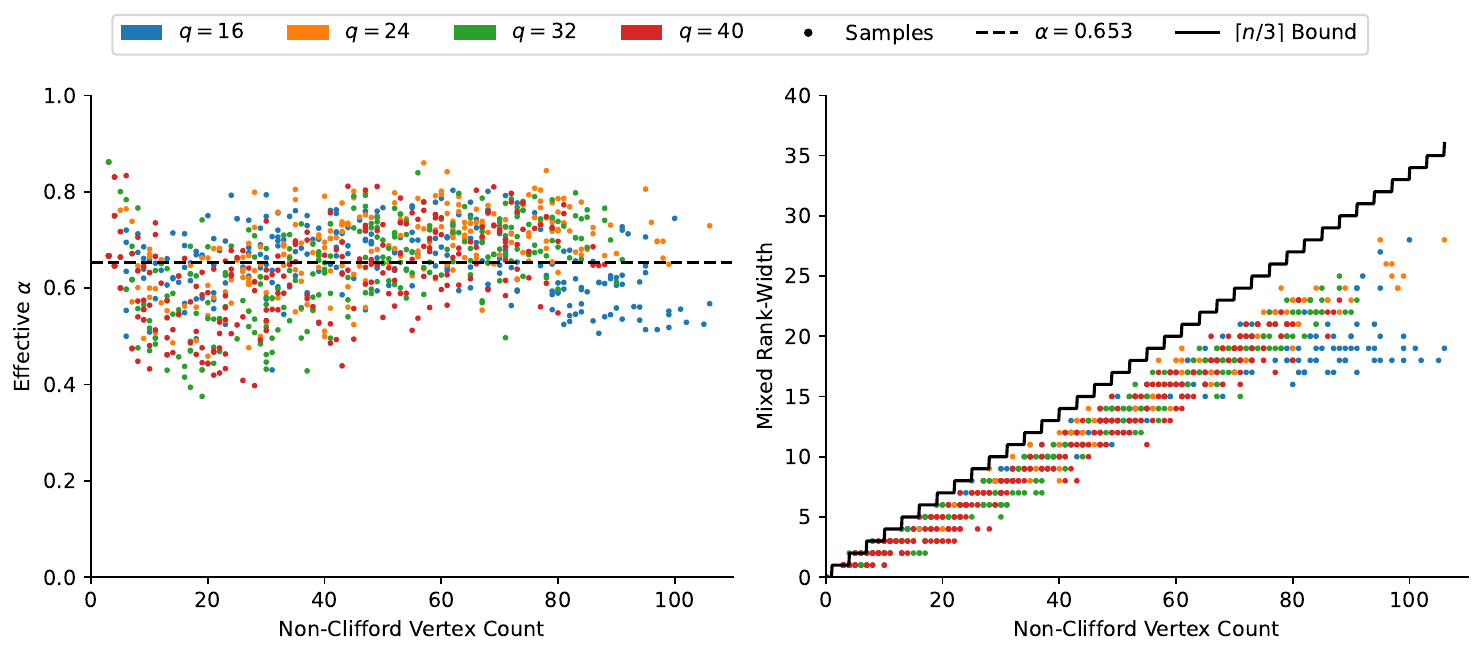}
    \caption{The mixed rank-width and computed $\alpha$ values to decompose a graph-like ZX-diagram obtained from simplifying a uniformly random Clifford+$R_Z$ circuit. We evaluated circuits with between $16$ and $40$ qubits, and between $100$ and $1000$ gates, with a fixed $20\%$ fraction of non-Clifford gates. We observe that the computed $\alpha$ is essentially independent of both qubit count and non-Clifford vertex count. The mean $\alpha \approx 0.653$ is marked in a dashed line. Moreover, we see that the mixed rank-width appears to grow linearly with non-Clifford vertex count, but does not saturate the upper bound. This suggests that the graphs obtained from this process do not look like Erd\H{o}s-R\'enyi graphs.}
    \label{fig:cliffordt}
\end{figure}

\begin{figure}
    \centering
    \hspace*{-\linewidth}\includegraphics[width=\linewidth]{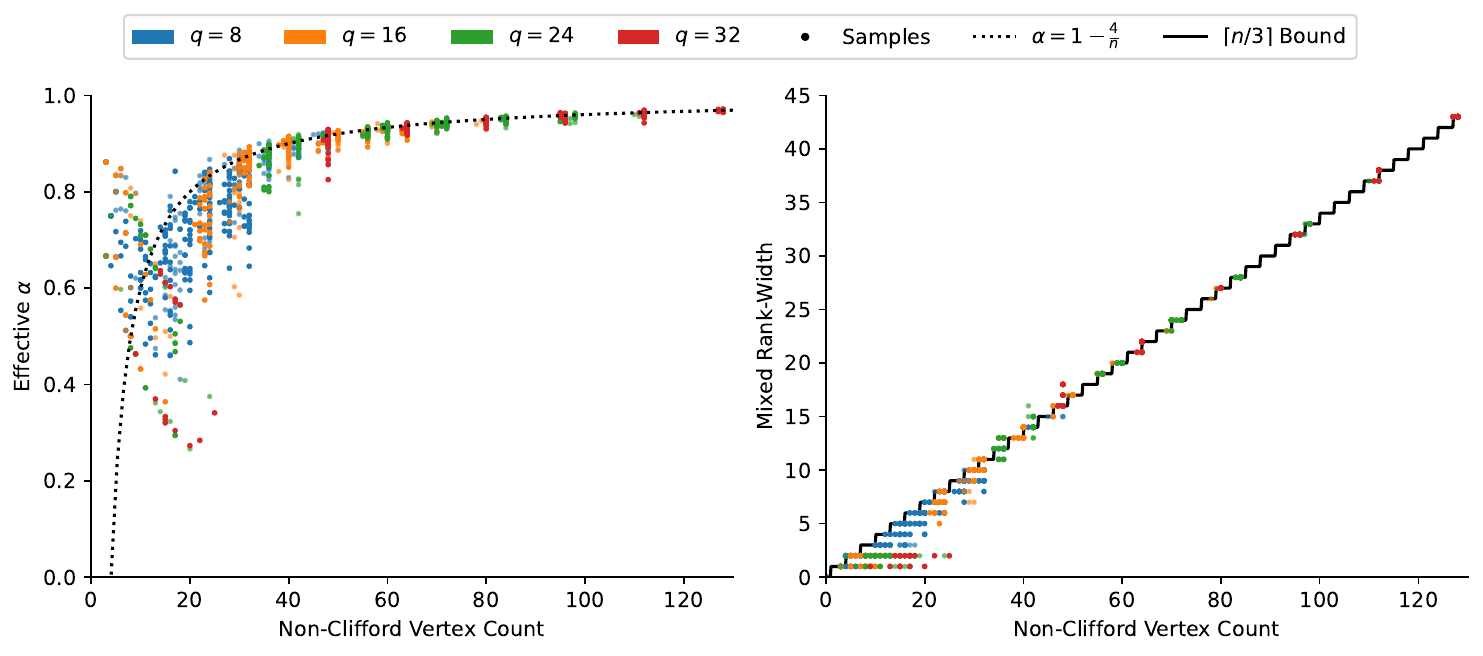}
    \caption{The mixed rank-width and computed $\alpha$ values to decompose a graph-like ZX-diagram obtained from simplifying a circuit composed of Pauli gadgets \cite{pauliGadget} with weight at most half the number of qubits. We sampled circuits with between $8$ and $32$ qubits and scaled the number of Pauli gadgets proportional to the number of qubits, by factors between one and four. We observe that, except for small graphs, the mixed rank-width upper bound is almost always saturated; the points that appear to exceed the upper bound have some remaining vertices that are Clifford, which are not taken into account in the bound. As the non-Clifford vertex count grows, the $\alpha$ value trends towards $1$. This indicates that these graphs have similar characteristics to Erd\H{o}s-R\'enyi graphs with intermediate values of $p$, and are essentially the hardest case for our methods. We observe a rough trend of $\alpha \approx 1 - \frac{4}{n}$, which is consistent with the theoretical value obtained by cutting every vertex except a constant number.}
    \label{fig:pauli}
\end{figure}


\newpage

\bibliography{main}

 \appendix

 \section{Appendix: An Introduction to the ZX-Calculus}
 \label{app:zxcalc}
 
 The \zxcalculus is a diagrammatic language similar to the familiar
quantum circuit notation.  A \emph{\zxdiagram} (or simply
\emph{diagram}) consists of \emph{wires} and \emph{spiders}.  Wires
entering the diagram from the left are \emph{inputs}; wires exiting to
the right are \emph{outputs}.  Given two diagrams we can compose them
by joining the outputs of the first to the inputs of the second, or
form their tensor product by simply stacking the two diagrams.

Spiders are linear operations that can have any number of inputs or outputs
wires.  There are two varieties, $Z$ spiders depicted as green dots: 
\begin{equation}
\begin{tikzpicture}
	\begin{pgfonlayer}{nodelayer}
		\node [style=Z phase dot] (0) at (0, 0) {$\alpha$};
		\node [style=none] (1) at (1.25, 0.825) {};
		\node [style=none] (2) at (-1.25, 0.825) {};
		\node [style=none] (3) at (-1.25, -0.75) {};
		\node [style=none] (4) at (1.25, -0.75) {};
		\node [style=none] (5) at (1.25, 0.5) {};
		\node [style=none] (6) at (-1.25, 0.5) {};
		\node [style=none, rotate=90] (7) at (-1, -0.125) {...};
		\node [style=none, rotate=90] (8) at (1, -0.125) {...};
	\end{pgfonlayer}
	\begin{pgfonlayer}{edgelayer}
		\draw [in=-141, out=0, looseness=0.75] (3.center) to (0);
		\draw [in=180, out=-39, looseness=0.75] (0) to (4.center);
		\draw [in=180, out=22, looseness=0.75] (0) to (5.center);
		\draw [in=180, out=39, looseness=0.75] (0) to (1.center);
		\draw [in=0, out=158, looseness=0.75] (0) to (6.center);
		\draw [in=141, out=0, looseness=0.75] (2.center) to (0);
	\end{pgfonlayer}
\end{tikzpicture} \ := \ \ketbra{0...0}{0...0} +
e^{i \alpha} \ketbra{1...1}{1...1} \label{eq:def-Z-spider}
\end{equation}
and $X$ spiders depicted as red dots:
\begin{equation}
\begin{tikzpicture}
	\begin{pgfonlayer}{nodelayer}
		\node [style=X phase dot] (9) at (0, 0) {$\alpha$};
		\node [style=none] (10) at (1.25, 0.825) {};
		\node [style=none] (11) at (-1.25, 0.825) {};
		\node [style=none] (12) at (-1.25, -0.75) {};
		\node [style=none] (13) at (1.25, -0.75) {};
		\node [style=none] (14) at (1.25, 0.5) {};
		\node [style=none] (15) at (-1.25, 0.5) {};
		\node [style=none, rotate=90] (16) at (-1, -0.125) {...};
		\node [style=none, rotate=90] (17) at (1, -0.125) {...};
	\end{pgfonlayer}
	\begin{pgfonlayer}{edgelayer}
		\draw [in=-141, out=0, looseness=0.75] (12.center) to (9);
		\draw [in=180, out=-39, looseness=0.75] (9) to (13.center);
		\draw [in=180, out=22, looseness=0.75] (9) to (14.center);
		\draw [in=180, out=39, looseness=0.75] (9) to (10.center);
		\draw [in=0, out=158, looseness=0.75] (9) to (15.center);
		\draw [in=141, out=0, looseness=0.75] (11.center) to (9);
	\end{pgfonlayer}
\end{tikzpicture} \ := \ \ketbra{+...+}{+...+} +
e^{i \alpha} \ketbra{-...-}{-...-}
\end{equation}
When the phase $\alpha$ is zero, we will omit it from the notation.
The diagram as a whole corresponds to a linear map built from the
spiders (and permutations) by the usual composition and tensor product
of linear maps.  As a special case, diagrams with no inputs represent
(unnormalised) state preparations.
For instance:
\begin{equation}\label{eq:state-ZX}
\begin{array}{rcccl}
  \begin{tikzpicture}
	\begin{pgfonlayer}{nodelayer}
		\node [style=Z dot] (0) at (-0.5, 0) {};
		\node [style=none] (1) at (0.5, 0) {};
	\end{pgfonlayer}
	\begin{pgfonlayer}{edgelayer}
		\draw (0) to (1.center);
	\end{pgfonlayer}
\end{tikzpicture} & = & \ket{0} + \ket{1}& \ =& \sqrt{2}\ket{+} \\[0.2cm]
  \begin{tikzpicture}
	\begin{pgfonlayer}{nodelayer}
		\node [style=X dot] (0) at (-0.5, 0) {};
		\node [style=none] (1) at (0.5, 0) {};
	\end{pgfonlayer}
	\begin{pgfonlayer}{edgelayer}
		\draw (0) to (1.center);
	\end{pgfonlayer}
\end{tikzpicture} & = & \ket{+} + \ket{-}& \ =& \sqrt{2}\ket{0} \\[0.2cm]
  \begin{tikzpicture}
	\begin{pgfonlayer}{nodelayer}
		\node [style=Z phase dot] (0) at (0, 0) {$\alpha$};
		\node [style=none] (1) at (-1.25, 0) {};
		\node [style=none] (2) at (1.25, 0) {};
	\end{pgfonlayer}
	\begin{pgfonlayer}{edgelayer}
		\draw (1.center) to (0);
		\draw (0) to (2.center);
	\end{pgfonlayer}
\end{tikzpicture} & = & \ketbra{0}{0} + e^{i \alpha} \ketbra{1}{1}&\ = & Z_\alpha \\[0.2cm]
  \end{array}
\end{equation}
Here the last one is the $Z_\alpha$ phase gate.

For convenience, special notation for the Hadamard gate is used:
\begin{equation}\label{eq:Hdef}
\hfill
\begin{tikzpicture}
	\begin{pgfonlayer}{nodelayer}
		\node [style=none] (0) at (7.3, 0) {};
		\node [style=none] (1) at (9.3, 0) {};
		\node [style=hadamard] (2) at (8.3, 0) {};
		\node [style=none] (3) at (6.3, 0) {$=:$};
		\node [style=none] (4) at (2.05, 0) {};
		\node [style=Z phase dot] (5) at (2.8, 0) {$\frac\pi2$};
		\node [style=X phase dot] (6) at (3.8, 0) {$\frac\pi2$};
		\node [style=Z phase dot] (7) at (4.8, 0) {$\frac\pi2$};
		\node [style=none] (8) at (5.55, 0) {};
		\node [style=box] (9) at (-3, 0) {H};
		\node [style=none] (10) at (-4.5, 0) {};
		\node [style=none] (11) at (-1.75, 0) {};
		\node [style=none] (12) at (-0.75, 0) {=};
		\node [style=none] (13) at (1, 0.1) {$e^{-i\pi/4}$};
	\end{pgfonlayer}
	\begin{pgfonlayer}{edgelayer}
		\draw (0.center) to (2);
		\draw (2) to (1.center);
		\draw (4.center) to (8.center);
		\draw (10.center) to (11.center);
	\end{pgfonlayer}
\end{tikzpicture}
\hfill
\end{equation}

Two diagrams are considered \emph{equal} when one can be deformed to
the other by moving the vertices around in the plane, bending,
unbending, crossing, and uncrossing wires, as long as the connectivity
and the order of the inputs and outputs is maintained. Equivalently, a
ZX-diagram can be considered as a graphical depiction of a tensor network,
as in e.g.~\cite{Penrose}. The Z- and X-spiders are symmetric tensors, and hence, like for other tensor networks of symmetric tensors, the interpretation of a ZX-diagram is unaffected by deformation. 

Quantum circuits can be straightforwardly translated into \zxdiagrams. 
In fact, as we can also represent state preparations and post-selections, 
ZX-diagrams with arbitrary angles are expressive enough to represent any linear map~\cite{CD2}. When we restrict the angles to multiples of $\pi/2$, the maps it represents correspond to Clifford maps: linear maps that can be expressed as a combination of stabiliser state preparations, Clifford unitaries, and stabiliser post-selections~\cite{BackensCompleteness}. 
Instead, restricting the angles to multiples of $\pi/4$ gives us the Clifford+$T$ \emph{fragment}, which corresponds to those linear maps that can be constructed from Clifford+$T$ unitaries together with state preparations and post-selections~\cite{ng_completeness_2018,SimonCompleteness}.

In addition to this extra flexibility which allows us to represent arbitrary linear maps, the real utility for ZX-diagrams comes from the set of rewrite rules they satisfy. This set of equations is called the \zxcalculus. Diagrams that can be transformed into each other using the rules of the ZX-calculus correspond to equal linear maps. We will only need a small number of rules:
\begin{equation}\label{eq:ZX-rules}
    \begin{tikzpicture}
	\begin{pgfonlayer}{nodelayer}
		\node [style=none] (24) at (-4.5, 2.5) {$=$};
		\node [style=none] (27) at (-1, 3.25) {};
		\node [style=none, rotate=90] (28) at (-3.375, 2.5) {...};
		\node [style=none] (29) at (-1, 1.625) {};
		\node [style=none, rotate=90] (30) at (-1.1, 2.5) {...};
		\node [style=none] (31) at (-3.5, 3.25) {};
		\node [style=Z phase dot] (32) at (-2.25, 2.5) {$\ \alpha\!+\!\beta\ $};
		\node [style=none] (34) at (-3.5, 1.625) {};
		\node [style=X dot] (35) at (1.125, 0.5) {};
		\node [style=none] (36) at (3, 0.5) {};
		\node [style=hadamard] (37) at (2.125, 0.5) {};
		\node [style=none] (38) at (4, 0.5) {=};
		\node [style=Z dot] (39) at (5.125, 0.5) {};
		\node [style=none] (40) at (6.125, 0.5) {};
		\node [style=Z phase dot] (41) at (-8.35, 0.5) {$\alpha$};
		\node [style=none] (42) at (-6.125, 0.5) {};
		\node [style=X phase dot] (43) at (-7.175, 0.5) {$b\pi$};
		\node [style=none] (44) at (-5.25, 0.5) {=};
		\node [style=Z phase dot] (45) at (-2, 0.5) {$\ (-1)^b\alpha~$};
		\node [style=none] (46) at (-0.25, 0.5) {};
		\node [style=none] (48) at (-4, 0.625) {$e^{ib\alpha}$};
		\node [style=none] (49) at (-2.75, -0.625) {$b\in\{0,1\}$};
		\node [style=none] (50) at (-5.5, 3.25) {};
		\node [style=none, rotate=90] (51) at (-8.5, 2.5) {...};
		\node [style=none] (52) at (-5.5, 1.625) {};
		\node [style=none, rotate=90] (53) at (-5.75, 2.5) {...};
		\node [style=none] (54) at (-8.75, 3.25) {};
		\node [style=Z phase dot] (55) at (-7.75, 2.5) {$\alpha$};
		\node [style=none] (56) at (-8.75, 1.625) {};
		\node [style=Z phase dot] (57) at (-6.5, 2.5) {$\beta$};
		\node [style=X dot] (58) at (2.25, 2.5) {};
		\node [style=none] (59) at (3, 1.75) {};
		\node [style=none] (60) at (3, 3.25) {};
		\node [style=Z phase dot] (61) at (1.125, 2.5) {$k\pi$};
		\node [style=none] (62) at (4, 2.5) {=};
		\node [style=none] (63) at (7.375, 1.875) {};
		\node [style=none] (64) at (7.375, 3.25) {};
		\node [style=Z phase dot] (65) at (6.375, 3.25) {$k\pi$};
		\node [style=Z phase dot] (66) at (6.375, 1.875) {$k\pi$};
		\node [style=none] (67) at (5.1, 2.5) {$\frac{1}{\sqrt{2}}$};
	\end{pgfonlayer}
	\begin{pgfonlayer}{edgelayer}
		\draw [style=simple, in=-120, out=0, looseness=0.75] (34.center) to (32);
		\draw [style=simple, in=-180, out=-60, looseness=0.75] (32) to (29.center);
		\draw [style=simple, in=180, out=45, looseness=0.75] (32) to (27.center);
		\draw [style=simple, in=0, out=135, looseness=0.75] (32) to (31.center);
		\draw (35) to (36.center);
		\draw (39) to (40.center);
		\draw (41) to (42.center);
		\draw (45) to (46.center);
		\draw [style=simple, in=-120, out=0, looseness=0.75] (56.center) to (55);
		\draw [style=simple, in=0, out=135, looseness=0.75] (55) to (54.center);
		\draw [in=180, out=45, looseness=0.75] (57) to (50.center);
		\draw [in=-180, out=-60, looseness=0.75] (57) to (52.center);
		\draw (55) to (57);
		\draw [in=-60, out=180] (59.center) to (58);
		\draw [in=60, out=180] (60.center) to (58);
		\draw (61) to (58);
		\draw (65) to (64.center);
		\draw (66) to (63.center);
	\end{pgfonlayer}
\end{tikzpicture}
\end{equation}
These are the \emph{spider-fusion} rule---that adjacent spiders of the same colour fuse together (which also holds for the X-spider)---and special cases of the \emph{colour-change} rule---that a Hadamard can be commuted through a spider to change its colour---and the \emph{$\pi$-copy} rules---that a $\pi$ phase can be commuted through the opposite colour~\cite{vandewetering2020zxcalculus}.

\end{document}